\documentclass[11pt,a4paper,twoside]{article}
\usepackage{a4wide, amssymb, amsmath, amsthm, graphics, comment, xspace, enumerate}
\usepackage{graphicx,multirow}
\usepackage[a4paper,colorlinks=true,citecolor=blue,urlcolor=blue,linkcolor=blue
                   , bookmarksopen=true,unicode=true,pdffitwindow=true]{hyperref}
\usepackage[section]{algorithm}
\usepackage[noend]{algpseudocode}
\usepackage{caption}
\usepackage{tabularx}
\usepackage{enumitem, multirow}
\usepackage[usenames,dvipsnames]{color}
\usepackage{alltt}
\usepackage{color}
\usepackage{listings}
\usepackage{cite}
\definecolor{string}{rgb}{0.7,0.0,0.0}
\definecolor{comment}{rgb}{0.13,0.54,0.13}
\definecolor{keyword}{rgb}{0.0,0.0,1.0}
\usepackage{amsmath,amsfonts,amssymb,subfigure}
\usepackage{tikz}
\usetikzlibrary{arrows}
\usetikzlibrary{decorations}
\usetikzlibrary{patterns}

\subfiguretopcaptrue
\tikzstyle{vtx}=[circle, inner sep= 0pt, minimum size= 1.2mm, fill]

\hypersetup{pdftitle={Efficient computation of trees with minimal atom-bond connectivity index}}

\newtheorem{te}{Theorem}[section]
\newtheorem{pro}[te]{Proposition}
\newtheorem{de}{Definition}[section]

\newtheorem{co}[te]{Corollary}
\newtheorem{lemma}[te]{Lemma}

\newcommand{\beq}{\begin{eqnarray}}
\newcommand{\eeq}{\end{eqnarray}}

\newcommand{\beqs}{\begin{eqnarray*}}
\newcommand{\eeqs}{\end{eqnarray*}}

\newcommand{\ABC}{{\rm ABC}}

\newcommand{\ds}{\displaystyle}
\allowdisplaybreaks

\begin{document}
\date{}
\title{ {On structural properties of trees \\with minimal atom-bond connectivity index} }
\maketitle
\begin{center}
{\large \bf  Darko Dimitrov}
\end{center}
\baselineskip=0.20in
\begin{center}
{\it Institut f\"ur Informatik, Freie Universit\"{a}t Berlin,
\\ Takustra{\ss}e 9, D--14195 Berlin, Germany} \\E-mail: {\tt darko@mi.fu-berlin.de}
\end{center}
\baselineskip=0.20in
\vspace{6mm}
\begin{abstract}
The {\em atom-bond connectivity  (ABC) index}  is a degree-based molecular descriptor,
that found chemical applications.
It is well known that among all connected graphs, 
the graphs with minimal ABC index are trees. 
A complete characterization of trees with minimal $ABC$ index is still an open problem.
In this paper, we present new structural properties of  trees with minimal
ABC index. Our main results reveal that trees with minimal ABC index do not
contain so-called {\em $B_k$-branches}, with $k \geq 5$,
and that they do not have more than four $B_4$-branches.
\end{abstract}
%
%
%
\medskip
%
%
%
%
\section[Introduction]{Introduction}

Description of the structure or shape of molecules is very helpful
in predicting activity and properties of molecules in complex experiments.
For that purpose, the molecular descriptors~\cite{tc-ndc-09} as
mathematical quantities are particularly useful.
Among the molecular descriptors, so-called topological indices \cite{db-tird-99} play a significant role.
%
The topological indices can be classified by the structural properties of graphs used for their calculation.
For example, the Wiener index~\cite{w-rppiams-1948} and the Balaban $J$ index~\cite{b-hddbti-82} are based on the distance of vertices in the respective
graph, the Estrada index \cite{e-c3dms-00} and the energy of a graph~\cite{g-eg-78} 
are based on the spectrum of the graph, the Zagreb group indices~\cite{GT}  and the Randi{\' c} connectivity index~\cite{r-cmb-1975}
depend on the degrees of vertices,  while the Hosoya index \cite{h-ti-1971} is calculated by the counting of non-incident
edges in a graph. On the other hand, there is a group of so-called information indices that are based on information functionals~\cite{b-ithiccs-83}. 
More about the information indices and the discriminative power of some established indices,
one can find in~\cite{dgv-iihdpg-12, dk-epge-12, da-hgem-11, fgd-ssdbti-13} and in the works cited therein.

 
Here, we consider a relatively new  topological index  which attracted a lot of attention in the last few years.
Namely, in 1998, Estrada et al. \cite{etrg-abc-98} proposed a new vertex-degree-based graph  topological index,
the {\em atom-bond connectivity (ABC) index},
and showed that it can be a valuable predictive tool in the study of the heat of formation in alikeness.
Ten years later Estrada~\cite{e-abceba-08} elaborated a novel quantum-theory-like justification for this topological index.
After that revelation, the interest of ABC-index has grown rapidly.
Additionaly, the physico-chemical applicability of the ABC index was confirmed and extended in several studies
\cite{as-abciic-10,cll-abcbsp-13, dt-cbfgaiabci-10, gg-nwabci-10, gtrm-abcica-12, k-abcibsfc-12, yxc-abcbsp-11}.

Let $G=(V, E)$ be a simple undirected graph of order $n=|V|$ and size $m=|E|$.
For $v \in V(G)$, the degree of $v$, denoted by $d(v)$, is the number of edges incident
to $v$.
For an edge $uv$ in $G$, let
\beq \label{eqn:000}
f(d(u), d(v))=\sqrt{\frac{d(u) +d(v)-2}{d(u)d(v)}}.
\eeq
Then the atom-bond connectivity index of $G$ is defined as
\beq \label{eqn:001}
\ABC(G)=\sum_{uv\in E(G)}f(d(u), d(v)). \nonumber
\eeq
As a new and well motivated graph invariant, the ABC index has attracted a lot of interest in the last several years both in 
mathematical and chemical research communities and numerous results and structural properties of ABC index  
were established~\cite{cg-eabcig-11, cg-abccbg-12, clg-subabcig-12, d-abcig-10, dgf-abci-11, dgf-abci-12, ftvzag-siabcigo-2011, fgv-abcit-09, gf-tsabci-12, gfi-ntmabci-12, llgw-pcgctmabci-13, vh-mabcict-2012, xz-etfdsabci-2012, xzd-abcicg-2011, xzd-frabcit-2010}.

The fact that adding an edge in a graph strictly increases its ABC index~\cite{dgf-abci-11} 
(or equivalently that deleting an edge in a graph strictly decreases its ABC index~\cite{cg-eabcig-11})  
has  the following two immediate consequences.

\begin{co}
Among all connected  graphs with $n$ vertices, the complete graph $K_n$ has maximal value of ABC index.
\end{co}

\begin{co}
Among all connected  graphs with $n$ vertices, the graph with minimal ABC index is a tree.
\end{co}

Although it is fairly easy to show that the star graph $S_n$
is a tree with maximal ABC index~\cite{fgv-abcit-09}, despite many attempts in the last years, it is still an open problem
the characterization of trees with minimal ABC index (also refereed as  minimal-ABC trees). 
The aim of this research is to make a step forward towards the full characterizations of minimal-ABC trees.

In Section~\ref{sec:known} we give an overview of already known structural properties of the minimal-ABC trees, while in 
Section~\ref{sec:new} we present a few new properties. In the appendix we present some simpler results that are used in
the proofs in Section~\ref{sec:new}.

In the sequel, we present an additional notation that will be used in the rest of the paper.
A tree is called a {\em rooted tree} if one vertex has been designated the {\em root}.
In a rooted tree, the  {\em parent} of a vertex is the vertex connected to it on the path to the root; every vertex except the root has a unique parent. 
A  {\em child} of a vertex $v$ is a vertex of which $v$ is the parent.
A vertex of degree one is a {\it pendant vertex}.
The {\it breadth-first search} is a graph search algorithm that begins at the root vertex and explores all its children vertices, beginning with the most right child and ending with the most left child. Then for each of those children, it explores their unexplored children vertices, and so on, until it finds the goal, or until all vertices are explored.

For the next two definitions, we adopt the notation from \cite{gfahsz-abcic-2013}.
Let $S_k=v_0 \, v_1 \dots v_k, v_{k+1}$,  $k \leq n-3$, be a sequence of vertices of a graph $G$
with  $d(v_0)>2$ and $d(v_i)=2$,  $i=1,\dots k-1$.
If $d(v_k)=1$, then $S_k$ is a {\it pendant path} of length $k+1$.
If $d(v_k) > 2$, then $S_k$ is an {\it internal path} of length $k$.




%
\section[Known structural properties of the minimal-ABC trees]
{Known structural properties of the minimal-ABC trees and some related results}\label{sec:known}


A thorough overview of the known structural properties of the minimal ABC-trees was given in~\cite{gfahsz-abcic-2013}.
In addition to the results mentioned  there,
we also present here the recently obtained related results that we are aware of.

To determine the minimal-ABC tress of order less than $10$ is a trivial task, and those
trees are depicted in Figure~\ref{fig-all-till-9}. To simplify the exposition in the rest of the paper, we assume
that the trees of interest are of order at least $10$.

\begin{figure}[h]
\begin{center}
\includegraphics[scale=0.75]{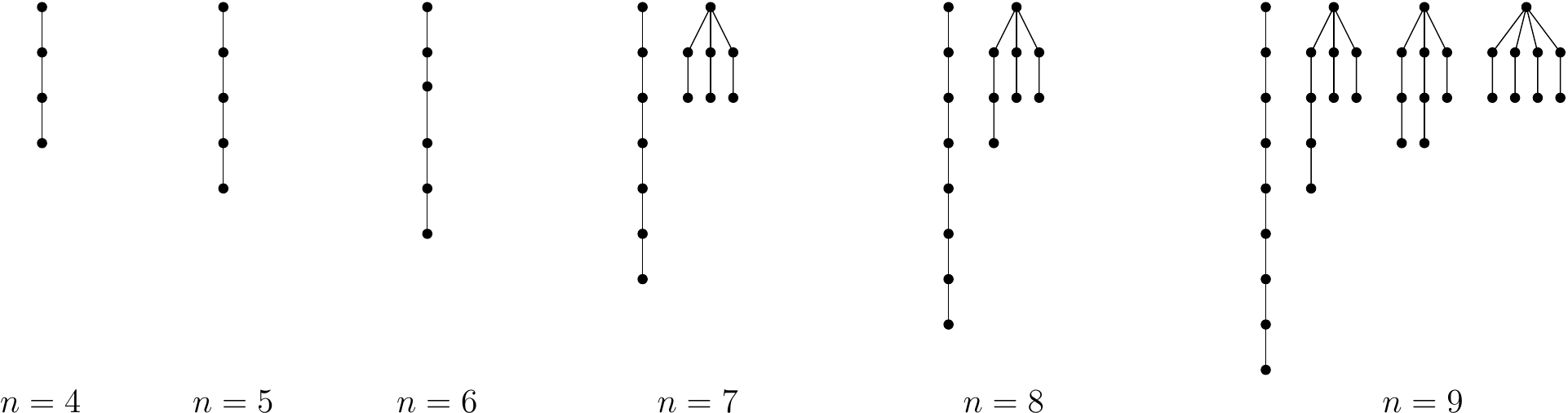}
\caption{Minimal-ABC trees of order $n$, $4 \leq n  \leq 9$.}
\label{fig-all-till-9}
\end{center}
\end{figure}

\smallskip
\noindent
In \cite{gfi-ntmabci-12}, Gutman, Furtula and Ivanovi{\' c}  obtained the following
results.

\begin{te}\label{thm-GFI-10}
An $n$-vertex tree with minimal ABC-index does not
contain internal paths of any length $k \geq 1$.
\end{te}

\begin{te}\label{thm-GFI-20}
An $n$-vertex tree with minimal ABC-index does not
contain pendant paths of length $k \geq 4$.
\end{te}

\noindent
An immediate, but important, consequence of Theorem~\ref{thm-GFI-10} is the next corollary.

\begin{co}\label{co-GFI-10}
 Let $T$ be a tree with minimal ABC index. Then the subgraph induced by the vertices of $T$ whose
degrees are greater than two is also a tree.
\end{co}

\noindent
An improvement of Theorem~\ref{thm-GFI-20} is the following result by Lin and Gao~\cite{llgw-pcgctmabci-13}.

\begin{te}\label{thm-LG-10}
Each pendant vertex of an $n$-vertex tree with minimal
ABC index belongs to a pendant path of length $k $, $2 \leq k \leq 3$.
\end{te}


\begin{te}[\cite{gfi-ntmabci-12}]\label{thm-GFI-30}
An $n$-vertex tree with minimal ABC-index contains at
most one pendant path of length $3$.
\end{te}

\noindent
Before we state the next important result, we consider the following  definition
of a {\em greedy tree} provided by Wang in~\cite{w-etwgdsri-2008}.

\begin{de}\label{def-GT}
Suppose the degrees of the non-leaf vertices are given, the greedy tree is achieved by the following `greedy algorithm':
\begin{enumerate}
\item Label the vertex with the largest degree as $v$ (the root).
\item Label the neighbors of $v$ as $v_1, v_2,\dots,$ assign the largest degree available to them such that $d(v_1) \geq d(v_2) \geq \dots$
\item Label the neighbors of $v_1$ (except $v$) as $v_{11}, v_{12}, \dots$ such that they take all the largest
degrees available and that $d(v_{11}) \geq d(v_{12}) \geq . . .$ then do the same for $v_2, v_3,\dots$
\item Repeat 3. for all newly labeled vertices, always starting with the neighbors of the labeled vertex with largest whose neighbors are not labeled yet.
\end{enumerate}
\end{de}
\noindent

\noindent
The following result  by Gan, Liu and You~\cite{gly-abctgds-12} characterizes the trees with minimal ABC index with prescribed degree sequences. 
The same result, using slightly different notation and approach, was obtained by
Xing and Zhou \cite{xz-etfdsabci-2012}.

\begin{te}\label{thm-DS}
Given the degree sequence, the greedy tree minimizes the ABC index.
\end{te}

\noindent
The next result was obtained in~\cite{gfahsz-abcic-2013}. Alternatively it can be obtained as a corollary of Theorem~\ref{thm-DS}. 
\begin{te}\label{thm-gutman-filomat}
If a minimal-ABC tree possesses three mutually adjacent vertices $v_1$, $v_2$, $v_3$, such that
$$
d(v_1) \geq d(v_2) \geq d(v_3),
$$
then $v_3$ must not be adjacent to both $v_1$ and $v_2$.
\end{te}

\noindent
To the best of our knowledge, the above mentioned results seems to be the only proven properties of the minimal ABC-trees.
%


For complete characterization of the minimal-ABC trees,
besides the theoretically proven properties, 
computer supported search can be of enormous help.
Therefore, 
we would like to mention in the sequel few related computational results.

A  first significant example of using computer search was done by Furtula et al.~\cite{fgiv-cstmabci-12}, 
where the trees with minimal ABC index of up to size of $31$ were computed,
and an initial conjecture of the general structure of the minimal-ABC trees was set.
There, a brute-force approach of generating all trees of a given order, 
speeded up by using  a distributed computing platform, was applied.
The plausible structural computational model  and its refined version presented there 
is based on the main assumption that  the minimal ABC tree posses a single {\it central vertex},
or said with other words, it is based  on the assumption that 
the vertices of a minimal ABC tree of degree $\geq 3$ induce a star graph.
This assumption was shattered by counterexamples presented in \cite{ahs-tmabci-13, ahz-ltmabci-13, d-ectmabci-2013}.
In this context, it is worth to mention that for a special class of trees, so-called {\em Kragujevac trees},
that are comprised of a central vertex and $B_k$-branches, $k \geq 1$ (see Figure~\ref{B_k-branches} for an illustration), the minimal-ABC tress were
fully characterized  by Hosseini, Ahmadi and Gutman~\cite{hag-ktmabci-14}.

In \cite{d-ectmabci-2013}  by considering only the degree sequences
of trees and some known structural properties of the trees with minimal ABC index
all trees with minimal ABC index of up to size of $300$ were computed.
%
\begin{figure}[h]
\begin{center}
\includegraphics[scale=0.75]{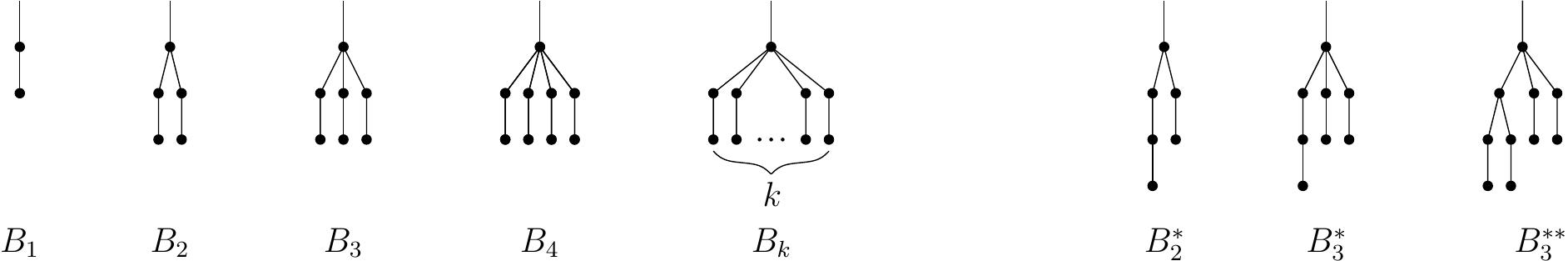}
\caption{$B_k$-branches. The vertex of a $B_k$-branch with degree $k+1$ is considered as the root of the branch.}
\label{B_k-branches}
\end{center}
\end{figure}
%
%
\section[Newly proven properties of the minimal-ABC trees]{New results} \label{sec:new}

By Theorem~\ref{thm-LG-10} and Corollary~\ref{co-GFI-10}, it follows that a minimal-ABC tree is comprised
of a tree $G$ to whose each pendant vertex a $B_k$-branch is attached.
Notice that if $G$ is just a single vertex, then the minimal-ABC tree is a Kragujevac tree.
In this section, we present new results considering the types of $B_k$-branches that a minimal-ABC tree cannot contain.
%
%
We start with the following result that will be used in the proof of Theorem~\ref{te-no5branches-10}.

\begin{pro} \label{pro-05}
A minimal-ABC tree does not contain 
a $B_1$-branch and a $B_{k \geq 5}$-branch
that have a common parent vertex.
\end{pro}

\begin{proof}

\smallskip
\noindent
Assume that there exist a minimal-ABC tree $G$ with a $B_1$-branch and a $B_{k \geq 5}$-branch 
that have a common parent vertex.  Denote that common vertex by $u$. Consider the transformation
$\mathcal{T}$ depicted in Figure~\ref{fig-common_vertex_P2_B5}. 
\begin{figure}[h!]
\begin{center}
\includegraphics[scale=0.750]{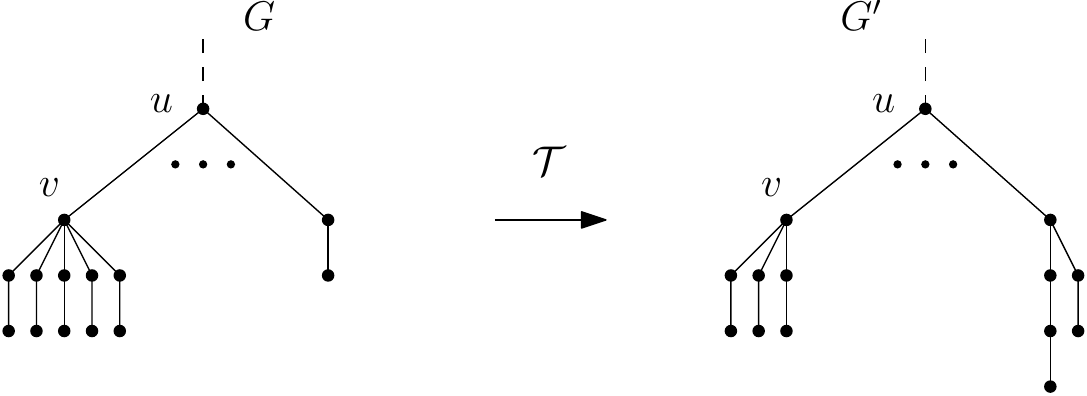}
\caption{Transforamation $\mathcal{T}$ from Proposition~\ref{pro-05}.}
\label{fig-common_vertex_P2_B5}
\end{center}
\end{figure}
The change of the ABC index after applying 
this transformation is
\beq \label{eq-pro-B5-10}
ABC(G') - ABC(G) &=& -f(d(u), d(v)) + f(d(u), d(v)-2) -f(d(u), 2) + f(d(u), 3) \nonumber \\
                              & = & g(d(u), d(v)).  \nonumber 
\eeq 
Here and in the rest of the paper, when we perform algebraic operations, we assume that
the degrees of the vertices can have real values.
By Proposition~\ref{appendix-pro-010}, the expression $-f(d(u), d(v)) + f(d(u), d(v)-2)$
decreases in $d(v)$, therefore $g(d(u), d(v))$ is maximal for $d(v)=6$.
The first derivative of $g(d(u), 6)$ is 
\beq \label{eq-pro-B5-15}
\frac{d\,g(d(u), 6)}{d\, d(u)} =  \frac{\frac{2 \sqrt{6}}{\sqrt{4+d(u)}}-\frac{3}{\sqrt{2+d(u)}}-\frac{\sqrt{3}}{\sqrt{1+d(u)}}}{6 d(u)^{3/2}},  \nonumber 
\eeq
which is equal to $0$ for $d(u)=31.3997$, where $g(d(u), 6)$ has its minimum.
Therefore, $g(d(u), 6)$ has its maximum at $d(u)=6$ or $d(u) \to \infty$, and 
\beq \label{eq-pro-B5-30}
ABC(G') - ABC(G) &\leq& \max  \left( g(6, 6),  \lim_{d(u) \to \infty} g(d(u), 6) \right) \nonumber  \\
&=&  \max  \left( -0.0331932, -0.0380048 \right) = -0.0331932. \nonumber   
\eeq
Thus, we have shown that after applying the transformation $\mathcal{T}$, the ABC index of $G$
decreases, which is a contradiction to the initial assumption that $G$ is a minimal-ABC tree.
\end{proof}

\begin{te} \label{te-no5branches-10}
A minimal ABC-tree does not contain a $B_k$-branch, where $k \geq 5$.
\end{te}

\begin{proof}
Let $G$ be a tree with minimal ABC index containg $B_{\geq5}$-branches.
We consider three cases with respect to the number of  $B_{\geq5}$-branches that $G$ may have:
$G$ has at least three  $B_{k \geq 5}$-branches, $G$ has two  $B_{k \geq 5}$-branches, and
$G$ has one  $B_{k \geq 5}$-branch.

\bigskip
\noindent
{\bf Case~$1.$} $G$ has at least three  $B_{k \geq 5}$-branches.

\smallskip
\noindent
If there are more than three $B_{k \geq 5}$-branches consider the last three with respect to the breadth-first search
of $G$ (recall that  by  Theorem~\ref{thm-DS} $G$ it is a greedy tree).
Denote the roots of those branches by $v_1$, $v_2$ and $v_3$.
Since they are roots of $B_{\geq5}$-branches, their degrees are at least $6$.
We assume that $d(v_1) \geq d(v_2) \geq d(v_3)$. Note that $v_1$, $v_2$ and $v_3$ can
have a common parent vertex, denoted here by $u_1$, or can have two different parent vertices, denoted by $u_1$ and $u_2$.
In the latter case, we assume that  $d(u_1) \geq d(u_2)$.
With respect to the number of parent vertices  of  $v_1$, $v_2$ and $v_3$, we distinguish three cases.

\smallskip
\noindent
{\bf Subcase~$1.1.$} $u_1$ is the parent vertex of $v_1$, and  $u_2$ is the parent vertex of $v_2$ and $v_3$.

\smallskip
\noindent
Apply the following transformation $\mathcal{T}_{11}$ to $G$:
from each of $v_1$, $v_2$ and $v_3$ cut an adjacent pendant path $P_2$, 
construct a $B_2^*$-branch and attach it to $u_2$.
An example of this case with an illustration of the transformation $\mathcal{T}_{11}$ is given in Figure~\ref{fig-case1_1}.
Observe that $u_1$ and $u_2$ can belong to different levels of $G$, for example as in Figure~\ref{fig-case1_1}.
\begin{figure}[h]
\begin{center}
\includegraphics[scale=0.750]{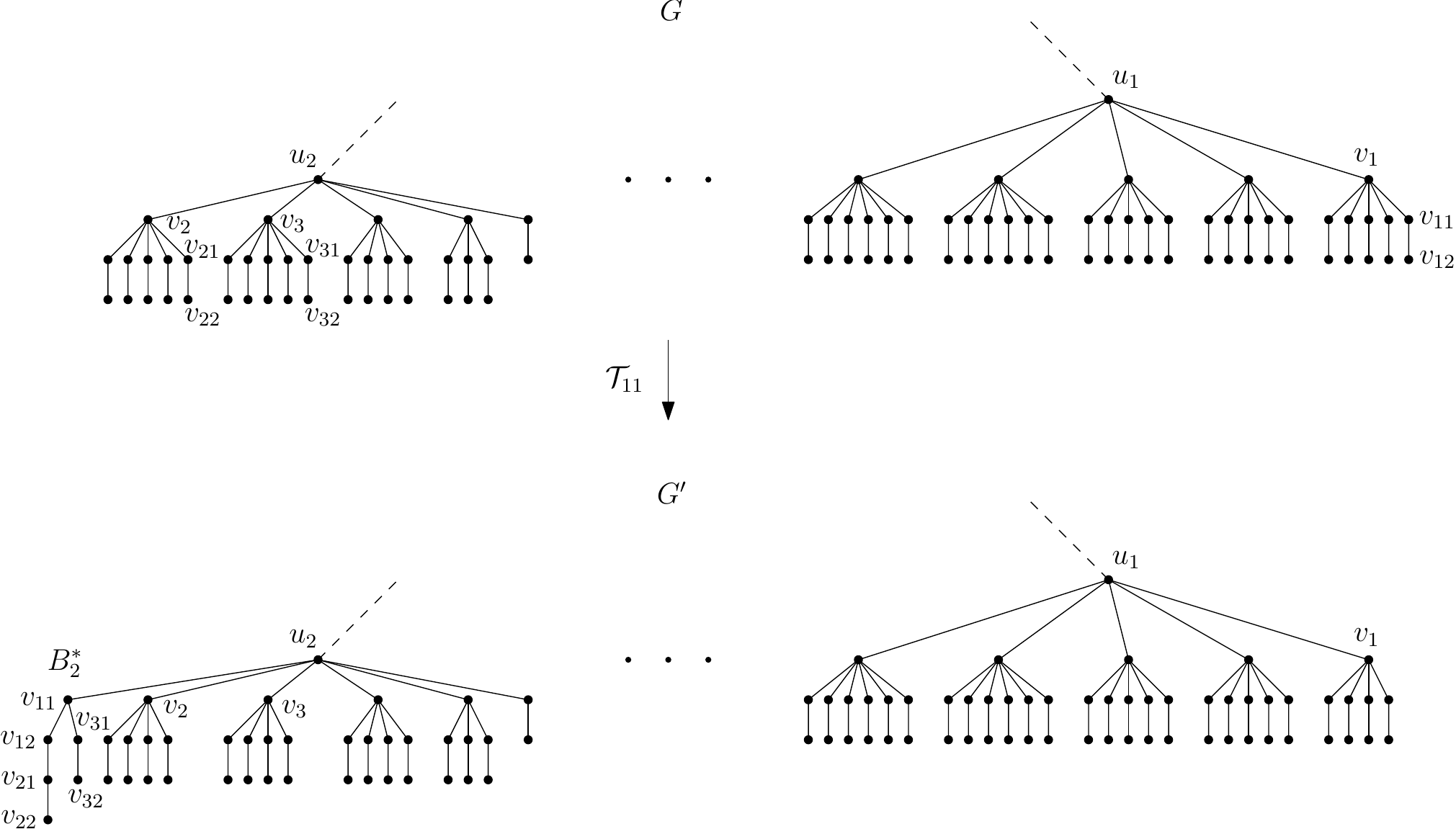}
\caption{Transforamation $\mathcal{T}_{11}$ from Subcase~$1.1.$ Only the relevant parts of $G$  and $G'$ are depicted.}
\label{fig-case1_1}
\end{center}
\end{figure}
After applying $\mathcal{T}_{11}$ the degrees of $v_1$, $v_2$ and $v_3$ decrease by one, while
the degrees of $u_2$ and $v_{11}$ increase by one. The degrees of the rest of the vertices of $G$,
including $u_1$, $v_{12}$, $v_{21}$, $v_{22}$, $v_{31}$ and $v_{32}$,  remain unchanged. 
%
%
\noindent
The change of the ABC index between $u_2$ and a vertex $w$, that is adjacent to $u_2$ and different than $v_2$ and $v_3$,
is:
\beq
-f(d(u_2), d(w)) + f(d(u_2)+1, d(w)) \nonumber
\eeq
which by Proposition~\ref{appendix-pro-020} is non-positive for $d(w) \geq 2$. 
%
The change of the ABC index between $u_2$  and $v_2$ is:
\beq \label{expression-30}
-f(d(u_2), d(v_2)) + f(d(u_2)+1, d(v_2)-1)
\eeq
By Proposition~\ref{appendix-pro-010}, the last expressions decreases in $d(v_2)$, i.e., it reaches it maximum for $d(v_2)=6$.
Thus the upper bound for (\ref{expression-30}) is 
\beq \label{expression-40}
-f(d(u_2), 6) + f(d(u_2)+1, 5) 
\eeq
Similarly, we obtain that (\ref{expression-40}) is an upper bound for  the change of ABC index between $u_2$  and $v_3$.
%
The change of ABC index between $u_1$ and $v_1$ is
\beq \label{expression-10}
-f(d(u_1), d(v_1)) + f(d(u_1), d(v_1)-1). 
\eeq
\noindent
By Proposition~\ref{appendix-pro-030}, the expression (\ref{expression-10}) increases in $d(u_1)$ and decreases in $d(v_1)$, thus it is maximal 
when $d(u_1)$ is maximal and $d(v_1)$ is minimal, i.e., $d(u_1)  \to \infty$  and $d(v_1)=6$.
Therefore,
\beq \label{expression-20}
&\lim_{ \substack{d(v_1)= 6 \\ d(u_1) \to \infty} }&   \left( -f(d(u_1), d(v_1)) + f(d(u_1), d(v_1)-1) \right)  =   \nonumber \\
&\lim_{ d(u_1) \to \infty}&  -\sqrt{\frac{d(u_1)+4}{ 6 d(u_1)}}+\sqrt{\frac{d(u_1)+3}{ 5 d(u_1)}} = 0.0389653  .
\eeq
is an upper bound for (\ref{expression-10}).

\noindent
Additionally, there is a change of the ABC index caused by  $v_{11}$
which is:
\beq \label{expression-50}
-f(d(v_1), d(v_{11})) + f(d(u_2)+1, d(v_{11}+1) ) =
-\sqrt{\frac{1}{ 2}} + f(d(u_2)+1, 3). 
\eeq
Thus, from  (\ref{expression-40}), (\ref{expression-20})  and (\ref{expression-50}), it follows that the total change of the ABC index of
$G$ after applying the transformation $\mathcal{T}_{11}$ is not larger than
\beq \label{expression-60}
0.0389653 + 2\left(  -f(d(u_2), 6) + f(d(u_2)+1, 5) \right)
-\sqrt{\frac{1}{ 2}} +f(d(u_2)+1, 3).
\eeq
 By Proposition~\ref{appendix-pro-050}, 
 $2\left(  -f(d(u_2), 6) + f(d(u_2)+1, 5) \right) +f(d(u_2)+1, 3)$  increases in $d(u_2)$, so the upper bound of
 the sum in (\ref{expression-60}) is 
\beq \label{expression-70}
\lim_{  d(u_2) \to \infty} && 0.0389653 + 2\left(  -f(d(u_2), 6) + f(d(u_2)+1, 5) \right)
-\sqrt{\frac{1}{ 2}} +f(d(u_2)+1, 3)  = \nonumber \\ 
 \lim_{  d(u_2) \to \infty} && 0.0389653 + 2\left(  -\sqrt{\frac{d(u_2)+4}{ 6 d(u_2)}} + \sqrt{\frac{d(u_2)+ 4}{ 5 (d(u_2)+1)}} \right)
-\sqrt{\frac{1}{ 2}} +\sqrt{\frac{d(u_2)+2}{ 3 (d(u_2)+1)}}  = \nonumber \\ 
&& -0.0128606  \nonumber.
\eeq
Thus, we have shown that  after applying the transformation $\mathcal{T}_{11}$
 the change of the ABC index of $G$ is strictly negative.


\bigskip
\noindent
{\bf Subcase~$1.2.$} $u_1$ is a parent vertex of $v_1$ and $v_2$, and  $u_2$ is a parent vertex of  $v_3$.

\smallskip
\noindent
Now, we apply the following transformation $\mathcal{T}_{12}$ to $G$:
from each of $v_1$, $v_2$ and $v_3$ cut an adjacent pendant path $P_2$, 
construct a $B_2^*$-branch and attach it to $u_1$.
An example of this case with an illustration of the transformation $\mathcal{T}_{12}$ is given in Figure~\ref{fig-case1_2}.

\begin{figure}[h]
\begin{center}
\includegraphics[scale=0.750]{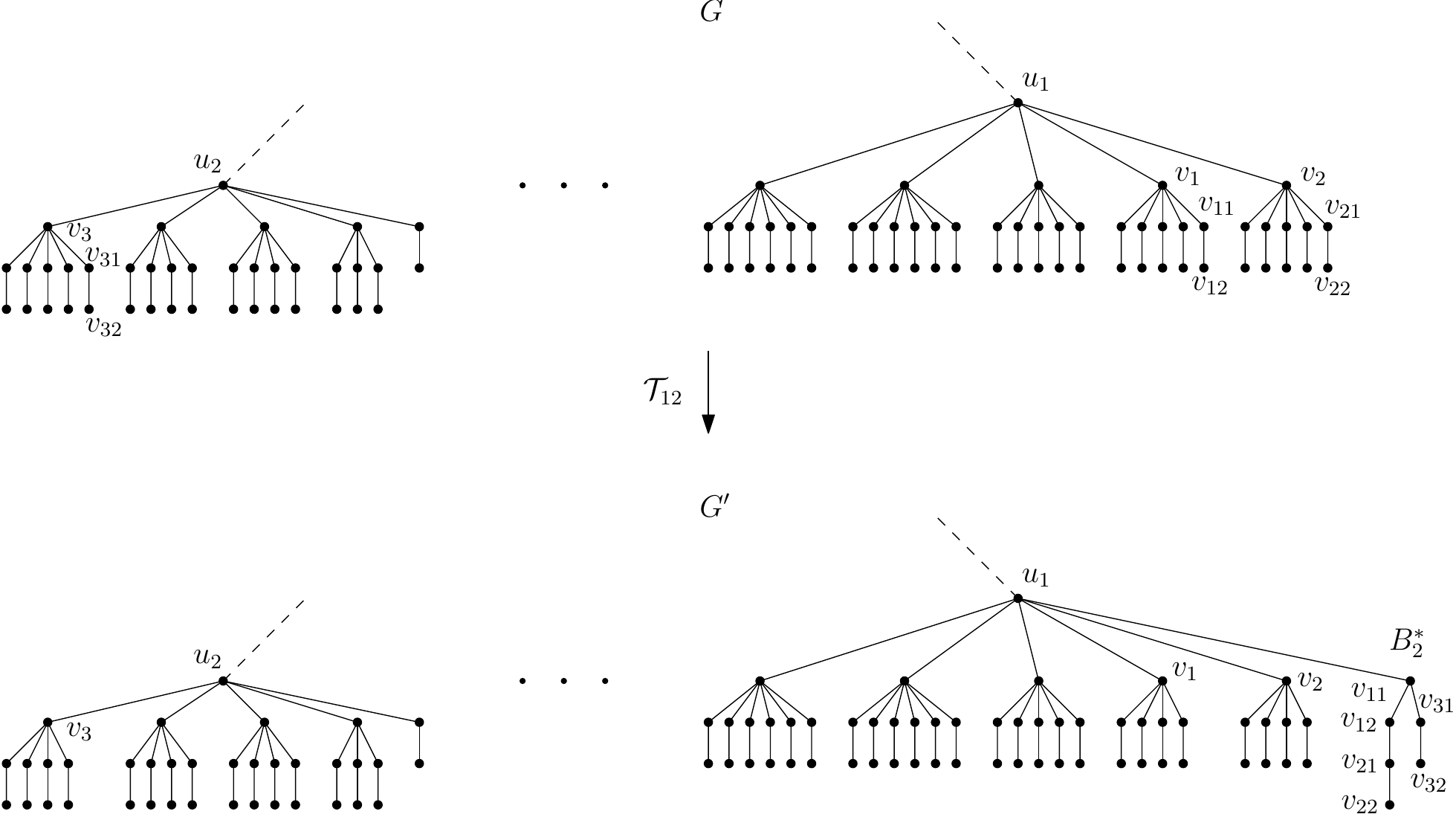}
\caption{Transforamation $\mathcal{T}_{12}$ from Subcase~$1.2.$ }
\label{fig-case1_2}
\end{center}
\end{figure}

After applying $\mathcal{T}_{12}$ the degrees of $v_1$, $v_2$ and $v_3$ decrease by one, while
the degrees of $u_1$ and $v_{11}$ increase by one. The degrees of the rest of the vertices of $G$  
remain unchanged. 
%
Since the analysis of the change of the ABC-index in  this subcase is almost identical to that of Subcase~$1.1.$ 
(the role of the vertices $u_1$ and $u_2$ are interchanged),
we omit the repetition of the detailed analysis, and just state the final upper bound on the change of  the 
ABC-index after applying the transformation $\mathcal{T}_{12}$, which is
\beq \label{expression-80}
0.0389653 + 2\left(  -f(d(u_1),6) + f(d(u_1)+1, 5) \right)
-\sqrt{\frac{1}{ 2}} +f(d(u_1)+1, 3) \leq -0.0128606. \nonumber
\eeq
%

\bigskip
\noindent
{\bf Subcase~$1.3.$} $u_1$ is a parent vertex of $v_1$, $v_2$, and  $v_3$.

\smallskip
\noindent
Similarly, as in the previous two subcases, we apply the transformation $\mathcal{T}_{13}$ to $G$:
from each of $v_1$, $v_2$ and $v_3$ cut an adjacent pendant path $P_2$, 
construct a $B_2^*$-branch and attach it to $u_1$.
An example of this case with an illustration of the transformation $\mathcal{T}_{13}$ is given in Figure~\ref{fig-case1_3}.
\begin{figure}[h]
\begin{center}
\includegraphics[scale=0.750]{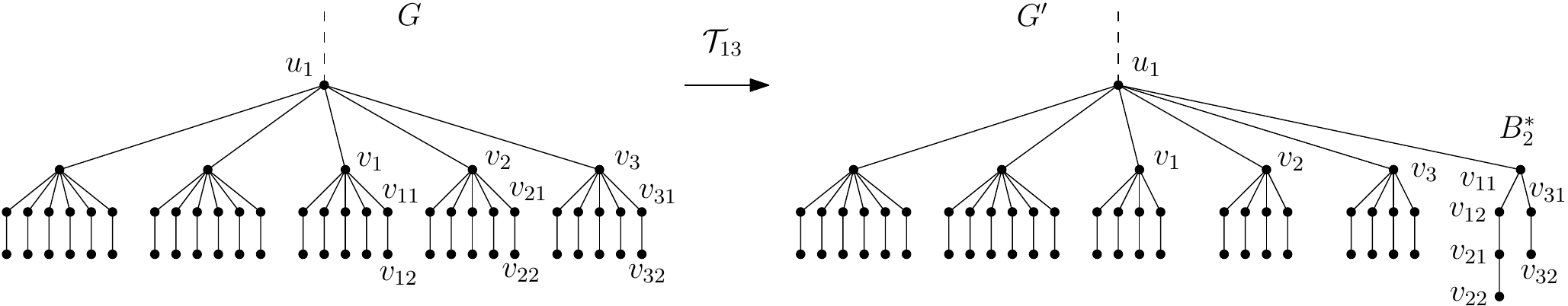}
\caption{Transforamation $\mathcal{T}_{13}$ from Subcase~$1.3.$}
\label{fig-case1_3}
\end{center}
\end{figure}
After applying $\mathcal{T}_{13}$ the degrees of $v_1$, $v_2$ and $v_3$ decrease by one, while
the degrees of $u_1$ and $v_{11}$ increase by one. The degrees of the rest of the vertices of $G$,
remain unchanged. 
%
\noindent
The change of the ABC index caused between $u_1$ and a vertex $w$, adjacent to $u_1$ and different than $v_1$ , $v_2$ and $v_3$,
is:
\beq
-f(d(u_1), d(w)) + f(d(u_1)+1, d(w))  \nonumber
\eeq
which by Proposition~\ref{appendix-pro-020} is non-positive for $d(w) \geq 2$, and strictly negative for $d(w) > 2$.
The change of the ABC index between $u_1$  and $v_1$ is:
\beq \label{expression-100}
f(d(u_1), d(v_1)) + f(d(u_1)+1, d(v_1)-1)
\eeq
By Proposition~\ref{appendix-pro-020}, the last expressions decreases in $d(v_1)$, i.e., it reaches it maximum for $d(v_1)=6$.
Thus the upper bound for (\ref{expression-100}) is 
\beq \label{expression-110}
-f(d(u_1), 6) + f(d(u_1)+1, 5) 
\eeq
Simillarly, we obtain that (\ref{expression-100}) is an upper bound for  the change of ABC index between 
$u_1$  and $v_2$, and $u_1$  and $v_3$.
Additionall, there is a change of the ABC index caused by  $v_{11}$
which is:
\beq \label{expression-120}
-f(d(v_1), d(v_{11})) + f(d(u_1)+1, d(v_{11}+1) ) =
-\sqrt{\frac{1}{ 2}} + f(d(u_2)+1, 3) 
\eeq
Thus, from (\ref{expression-110})  and (\ref{expression-120}), it follows that the total change of the ABC index of
$G$ after applying the transformation $\mathcal{T}_{11}$ is at most
\beq \label{expression-130}
 3\left(  -f(d(u_1), 6) + f(d(u_1)+1, 5) \right)
-\sqrt{\frac{1}{ 2}} +f(d(u_1)+1, 3).
\eeq
 By Proposition~\ref{appendix-pro-050}, 
 $3\left(  -f(d(u_1), 6) + f(d(u_1)+1, 5) \right) +f(d(u_1)+1, 3)$  increases in $d(u_1)$, so the upper bound of
 the sum in (\ref{expression-130}) is 
%
\beq \label{expression-140}
\lim_{  d(u_2) \to \infty} &&   3\left(  -f(d(u_1), 6) + f(d(u_1)+1, 5) \right)
-\sqrt{\frac{1}{ 2}} +f(d(u_1)+1, 3)
 = \nonumber \\ 
 \lim_{  d(u_1) \to \infty} && 3\left(  -\sqrt{\frac{d(u_1)+4}{ 6 d(u_1)}} + \sqrt{\frac{d(u_1)+ 4}{ 5 (d(u_1)+1)}} \right)
-\sqrt{\frac{1}{ 2}} +\sqrt{\frac{d(u_1)+2}{ 3 (d(u_1)+1)}}  = -0.0128606. \nonumber 
\eeq
Thus, we have shown that  after applying the transformation $\mathcal{T}_{13}$
 the change of the ABC index of $G$ is strictly negative.

Thus, applying iteratively $\mathcal{T}_{11}$, $\mathcal{T}_{12}$ and $\mathcal{T}_{13}$, we  obtain a tree
that have at most two $B_{k \geq 5}$-branches and has smaller ABC-index than $G$.

Notice that $G'$, a tree obtain after applying $\mathcal{T}_{11}$,  $\mathcal{T}_{12}$ or $\mathcal{T}_{13}$,
is not necessarily a minimal ABC-tree,
since it may not be a greedy tree (as it is the case with the examples in Figures~\ref{fig-case1_1}, ~\ref{fig-case1_2}, and ~\ref{fig-case1_3}).
In that case, one can transform $G'$ into a minimal ABC-tree with a same degree sequence as $G'$ by Theorem~\ref{thm-DS}.

\bigskip
\noindent
In the following two cases (Case~$2$ and $3$), we will take in the account the result from Proposition~\ref{pro-05}
that in a minimal-ABC tree there is no vertex that has simultaneously a $B_1$-branch and a $B_{k \geq 5}$-branch as its children.

\bigskip
\noindent
{\bf Case~$2.$} 
 $G$ has two  $B_{k \geq 5}$-branches.
 
\smallskip
\noindent
 Denote the root vertices of the $B_{k \geq 5}$-branches by $v_1$ and $v_2$.
 The vertices  $v_1$ and $v_2$ may have different parent vertices, denoted by $u_1$ and $u_2$, 
 or they may have the same parent vertex, denoted by $u_1$.
 The two cases we analyze separately.
 We assume that $d(u_1) \geq d(u_2)$ and $d(v_1) \geq d(v_2)$.

\bigskip
\noindent
{\bf Subcase~$2.1.$} The vertex $v_1$ is a child $u_1$ and the vertex $v_2$ is a child  of $u_2$.


\smallskip
\noindent
{\bf Subcase~$2.1.1.$}  The vertex $u_2$  has a child $v_3$ of degree $3$ or $4$.

\noindent
Apply the following transformation $\mathcal{T}_{211}$ to $G$:
from $v_2$ cut an adjacent pendant path $P_2$, 
and attach it to $v_3$.
An example of this case with an illustration of the transformation $\mathcal{T}_{211}$, when $d(v_3)=3$, is given in Figure~\ref{fig-case2_11}.
\begin{figure}[h]
\begin{center}
\includegraphics[scale=0.750]{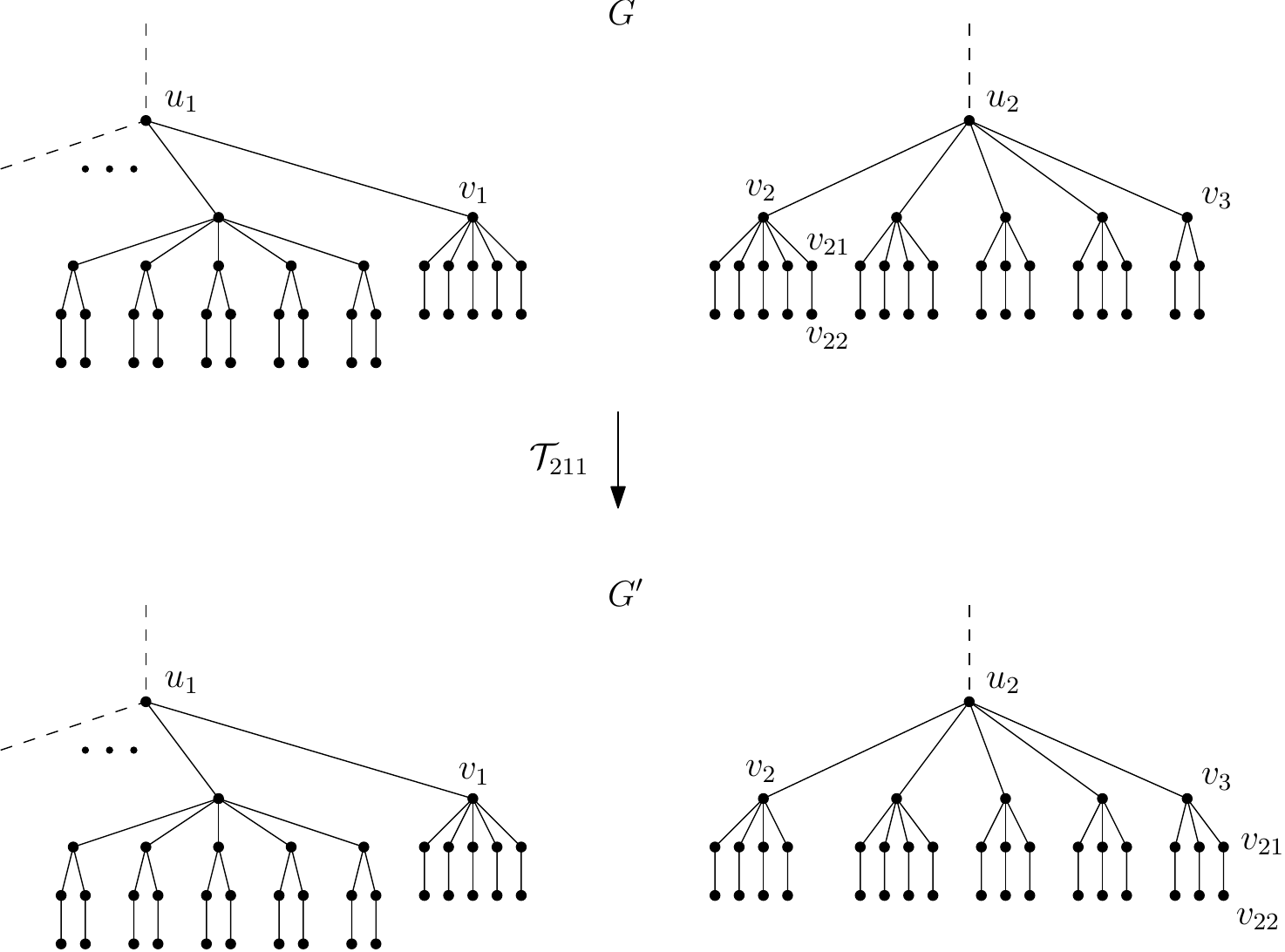}
\caption{Transforamation $\mathcal{T}_{211}$ from Subcase~$2.1.1.$}
\label{fig-case2_11}
\end{center}
\end{figure}
After applying $\mathcal{T}_{211}$ the degree of $v_2$ decreases by one, while
the degree of $v_3$ increases by one. The degrees of the rest of the vertices of $G$  remain unchanged. 
The change of the ABC index is
\beq \label{expression-150}
  -f(d(u_2), d(v_2)) + f(d(u_2), d(v_2)-1) 
   -f(d(u_2), d(v_3)) + f(d(u_2), d(v_3)+1).
\eeq
By Proposition~\ref{appendix-pro-030},  $-f(d(u_2), d(v_2)) + f(d(u_2), d(v_2)-1)$
decreases in  $d(v_2)$, thus, the expression (\ref{expression-150}) is maximal for $d(v_2)=6$.
Due to the symmetry of the function $f$, 
$-f(d(u_2), d(v_3)) + f(d(u_2), d(v_3)+1) = -f(d(v_3), d(u_2)) + f(d(v_3)+1,d(u_2))$,
and by Proposition~\ref{appendix-pro-020}, it increases in  $d(v_3)$. Since $d(v_3)$ is $3$ or $4$, we  take $d(v_3)=4$ and
\beq \label{expression-160}
  &&-f(d(u_2), 6) + f(d(u_2), 5) 
   -f(d(u_2), 4) + f(d(u_2), 5), \quad \text{or} \nonumber \\
   &&-f(4, d(u_2)) + f(5, d(u_2)) -(-f(d(5, u_2)) + f(6, d(u_2)))
\eeq
as an upper bound on (\ref{expression-150}).
By Proposition~\ref{appendix-pro-020}, it follows that
$-f(d(u_2), 5) + f(d(u_2), 6) > -f(d(u_2), 4) + f(d(u_2), 5)$, and thus, the expression (\ref{expression-160}),
and consequently (\ref{expression-150}), are negative.

\smallskip
\noindent
{\bf Subcase~$2.1.2.$}  The children of vertex $u_2$, different than $v_2$,  have degrees $5$.

\noindent
In this case, apply the following transformation $\mathcal{T}_{212}$ to $G$:
from $v_1$, $v_2$, and two children vertices of $u_2$ (denoted by $v_3$ and $v_4$), 
cut an adjacent pendant path $P_2$, form a $B_3^*$-branch
and attach it to $u_2$.
An example of this case with an illustration of the transformation $\mathcal{T}_{212}$ is given in Figure~\ref{fig-case2_12}.
\begin{figure}[h]
\begin{center}
\includegraphics[scale=0.750]{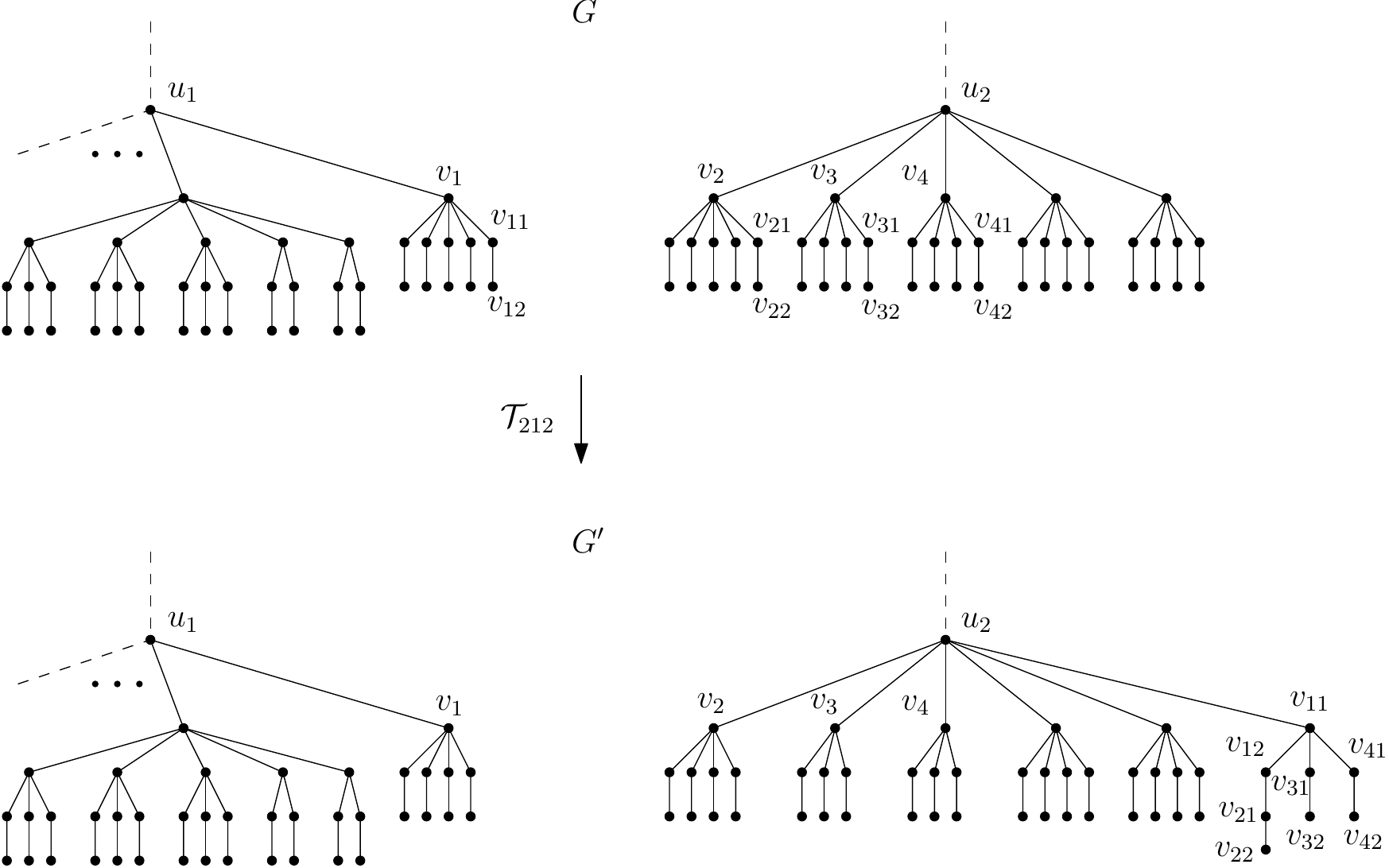}
\caption{Transforamation $\mathcal{T}_{212}$ from Subcase~$2.1.2.$}
\label{fig-case2_12}
\end{center}
\end{figure}
After applying $\mathcal{T}_{212}$ the degrees of $v_1$,$v_2$,  $v_3$ and $v_4$ decrease by one, 
the degree of $u_2$ increases by one, and the degree of one child vertex of $v_1$, denoted by $v_{11}$, increases by two. 
The degrees of the rest of the vertices of $G$  remain unchanged. 
The change of the ABC index is
\beq \label{expression-170}
 && -f(d(u_1), d(v_1)) + f(d(u_1), d(v_1)-1) 
   -f(d(u_2), d(v_2)) + f(d(u_2), d(v_2)-1)  \nonumber \\ 
 &&   +2(-f(d(u_2), 5) + f(d(u_2), 4) )
   -f(d(v_1), 2) + f(d(u_2), 4). 
\eeq
By Proposition~\ref{appendix-pro-030},  $-f(d(u_1), d(v_1)) + f(d(u_1), d(v_1)-1) $ 
(resp. $-f(d(u_2), d(v_2)) + f(d(u_2), d(v_2)-1)$) decreases in $d(v_1)$ (resp. $d(v_2)$),
therefore, (\ref{expression-170}) is maximal for $d(v_1)=d(v_2)=6$.
Also by Proposition~\ref{appendix-pro-030},  $-f(d(u_1), d(v_1)) + f(d(u_1), d(v_1)-1)$ increases
in $d(u_1)$, and   (\ref{expression-170}) is maximal for $d(u_1) \to \infty$.
Thus,
\beq \label{expression-180}
&& \lim_{d(u_1) \to \infty} (-f(d(u_1), 6) + f(d(u_1), 5)) = 0.0389653,  \quad \text{and} \nonumber \\
 &&0.0389653 -f(d(u_2), 6) + f(d(u_2), 5)  +2(-f(d(u_2), 5) + f(d(u_2), 4) ) -\sqrt{\frac{1}{2}}+ f(d(u_2), 4)   \nonumber \\
 =&-&f(d(u_2), 6) + f(d(u_2), 5)  +2(-f(d(u_2), 5) + f(d(u_2), 4) ) + f(d(u_2), 4)  - 0.668141
\eeq
is an upper bound on (\ref{expression-170}).
The first derivative of the function
$
 g(d(u_2))=-f(d(u_2), 6) + f(d(u_2), 5)+2(-f(d(u_2), 5) + f(d(u_2), 4)) + f(d(u_2), 4)
$
after a simplification is
\beq \label{expression-200}
\ds \frac{d \, g(d(u_2))}{d \ d(u_2)} = \frac{-\frac{45}{\sqrt{\frac{2+d(u_2)}{d(u_2)}}}+\frac{9 \sqrt{5}}{\sqrt{\frac{3+d(u_2)}{d(u_2)}}}+\frac{10 \sqrt{6}}{\sqrt{\frac{4+d(u_2)}{d(u_2)}}}}{30 d(u_2)^2}. \nonumber
\eeq
It holds that
\beq \label{expression-210}
-\frac{45}{\sqrt{\frac{2+d(u_2)}{d(u_2)}}}+\frac{9 \sqrt{5}}{\sqrt{\frac{3+d(u_2)}{d(u_2)}}}+\frac{10 \sqrt{6}}{\sqrt{\frac{4+d(u_2)}{d(u_2)}}}
< -\frac{45}{\sqrt{\frac{3+d(u_2)}{d(u_2)}}}+\frac{9 \sqrt{5}}{\sqrt{\frac{3+d(u_2)}{d(u_2)}}}+\frac{10 \sqrt{6}}{\sqrt{\frac{3+d(u_2)}{d(u_2)}}} < 0, \nonumber
\eeq
form which it follows that  the expression (\ref{expression-180}) decreases in $d(u_2)$, and reaches it maximum for $d(u_2)=6$.
So, the upper bound on (\ref{expression-180}), and therefore for (\ref{expression-170}) is
$$
-f(6,6)+f(6,5)+2(-f(6,5)+f(6,4))+f(6,4)-0.668141 =-0.0108595.
$$
Thus, we have shown that the change of the ABC index in this case is negative.

\bigskip
\noindent
{\bf Subcase~$2.2.$} $v_1$ and $v_2$ are children of same vertex  $u_1$.

\smallskip
\noindent
{\bf Subcase~$2.2.1.$}  The vertex $u_1$  has a child $v_3$ of degree $3$ or $4$.

\noindent
Here, we apply the following transformation $\mathcal{T}_{221}$ to $G$:
from $v_1$ cut an adjacent pendant path $P_2$, 
and attach it to $v_3$.
An example of this case with an illustration of the transformation $\mathcal{T}_{221}$ is given in Figure~\ref{fig-case2_21}.
\begin{figure}[h]
\begin{center}
\includegraphics[scale=0.750]{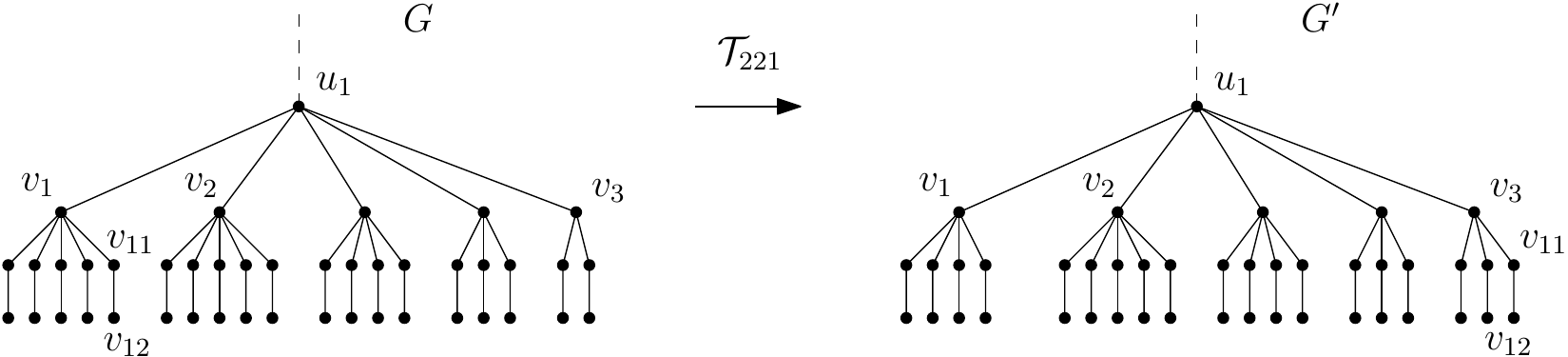}
\caption{Transforamation $\mathcal{T}_{221}$ from Subcase~$2.2.1.$}
\label{fig-case2_21}
\end{center}
\end{figure}
After applying $\mathcal{T}_{221}$ the degree of $v_1$  decreases by one, while
the degrees of $v_3$ increases by one. The degrees of the rest of the vertices of $G$  remain unchanged. 
The change of the ABC index is
\beq \label{expression-220}
  -f(d(u_1), d(v_1)) + f(d(u_1), d(v_1)-1) 
   -f(d(u_1), d(v_3)) + f(d(u_1), d(v_3)+1).
\eeq
Observe that (\ref{expression-220}) is very similar to (\ref{expression-150}), only the role of
$d(u_1)$ and $d(u_2)$ are interchange. Therefore, we will omit the analysis in here, and 
just state the final conclusion that (\ref{expression-220}) is always negative.

\smallskip
\noindent
{\bf Subcase~$2.2.2.$}  The vertex $u_1$  does not have a child of degree $3$ and $4$.

\noindent
By Proposition~\ref{pro-05}, $u_1$  does not have a child of degree $2$, i.e.,
all children of $u_1$, except $v_1$ and $v_2$, have degrees $5$.
Now, apply the following transformation $\mathcal{T}_{222}$ to $G$:
from $v_1$, $v_2$, and two children vertices of $u_1$ (named by $v_3$ and $v_4$), 
cut an adjacent pendant path $P_2$, form a $B_3^*$-branch
and attach it to $u_1$.
An example of this case with an illustration of the transformation $\mathcal{T}_{222}$ is given in Figure~\ref{fig-case2_22}.
\begin{figure}[h]
\begin{center}
\includegraphics[scale=0.750]{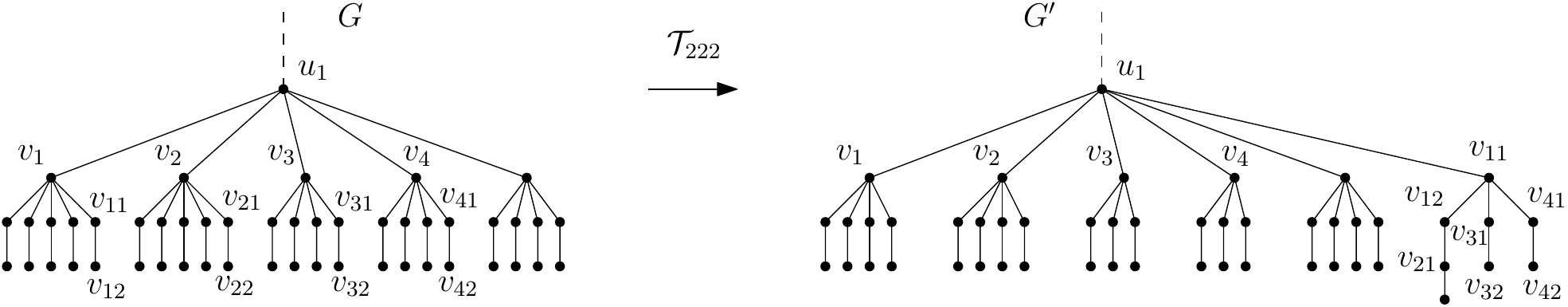}
\caption{Transforamation $\mathcal{T}_{222}$ from Subcase~$2.2.2.$}
\label{fig-case2_22}
\end{center}
\end{figure}
%
After this transformation the degrees of $v_1$,$v_2$,  $v_3$ and $v_4$ decrease by one, 
the degree of $u_1$ increases by one, and the degree of one child vertex of $v_1$, denoted by $v_{11}$, increases by two. 
The degrees of the rest of the vertices of $G$  remain unchanged. 
The change of the ABC index here is
\beq \label{expression-230}
 && -f(d(u_1), d(v_1)) + f(d(u_1), d(v_1)-1) 
   -f(d(u_1), d(v_2)) + f(d(u_1), d(v_2)-1)  \nonumber \\ 
 &&   +2(-f(d(u_1), 5) + f(d(u_1), 4) )
   -f(d(v_1), 2) + f(d(u_1), 4). 
\eeq

\noindent
By Proposition~\ref{appendix-pro-030},  $-f(d(u_1), d(v_1)) + f(d(u_1), d(v_1)-1) $ 
(resp. $-f(d(u_1), d(v_2)) + f(d(u_1), d(v_2)$ $-1)$) decreases in $d(v_1)$ (resp. $d(v_2)$),
therefore, (\ref{expression-230}) is maximal for $d(v_1)=d(v_2)=6$.
Thus,
\beq \label{expression-240}
g(d(u_1)) &= &2(-f(d(u_1), 6) + f(d(u_1), 5))  +2(-f(d(u_1), 5) + f(d(u_1), 4) ) \nonumber \\ 
&& + f(d(u_1), 4)  -\sqrt{\frac{1}{2}}
\eeq
is an upper bound on (\ref{expression-230}).
The first derivative of $g(d(u_1)) $
after a simplification is
\beq \label{expression-250}
\ds \frac{d \, g(d(u_1))}{d  \, d(u_1)} = 
\frac{-\frac{9}{\sqrt{\frac{2+u_1}{u_1}}}+\frac{4 \sqrt{6}}{\sqrt{\frac{4+u_1}{u_1}}}}{6 u_1^2}, \nonumber
\eeq
which is negative for $u_1 < 8.8$, and positive for $u_1 > 8.8$, i.e.,
$g(d(u_1))$ decreases in $u_1$ when $u_1 \in \{6, 7, 8 \}$, and increases in $u_1$ when $u_1 \in [9,\infty)$.
Thus the upper bound on the expression $g(d(u_1))$, and therefore on (\ref{expression-230}) is
$$
\max \left(g(6), \lim_{d(u_1) \to \infty} g(d(u_1))\right) = \max \left(-0.0291485,  -0.0236034 \right) = -0.0236034,
$$
\noindent
Thus, we have shown that the change of the ABC index also in this case is negative.

\bigskip
\noindent
{\bf Case~$3.$} 
$G$ has one  $B_{k \geq 5}$-branch.

\noindent
We denote by  $v_1$ the root of the $B_{k \geq 5}$-branch, and by $u_1$ the parent vertex of $v_1$.
By Proposition~\ref{pro-05}, it follows that $u_1$  does not have a child of degree $2$, i.e.,
all children of $u_1$, except $v_1$, have degrees $3$, $4$, or $5$.

\smallskip
\noindent
{\bf Subcase~$3.1.$} The vertex $u_1$ has a children of degrees $3$ or $4$.

\noindent
Let $v_2$ be such a child of $u_1$ of degree $3$ or $4$.
Apply the following transformation $\mathcal{T}_{31}$ to $G$:
from $v_1$ cut an adjacent pendant path $P_2$, 
and attach it to $v_2$.
An example of this case with an illustration of the transformation $\mathcal{T}_{31}$, when $d(v_3)=3$, is given in Figure~\ref{fig-case3_1}.
\begin{figure}[h]
\begin{center}
\includegraphics[scale=0.750]{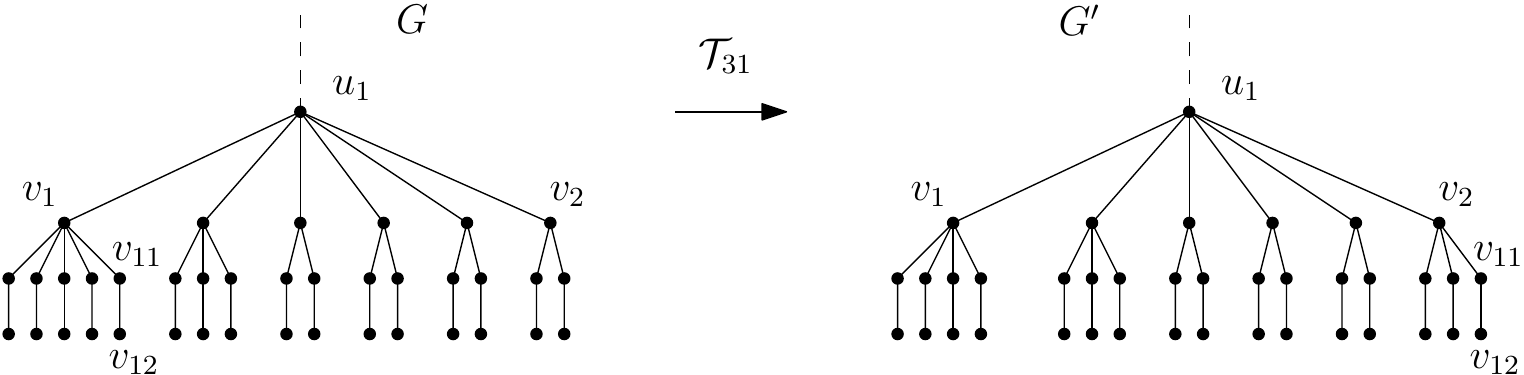}
\caption{Transforamation $\mathcal{T}_{31}$ from Subcase~$3.1.$}
\label{fig-case3_1}
\end{center}
\end{figure}
After applying $\mathcal{T}_{31}$ the degree $v_1$ decrease by one, while
the degree of $v_2$ increases by one. The degrees of the rest of the vertices of $G$  remain unchanged. 
The change of the ABC index is
\beq \label{expression-260}
  -f(d(u_1), d(v_1)) + f(d(u_1), d(v_1)-1) 
   -f(d(u_1), d(v_2)) + f(d(u_1), d(v_2)+1).
\eeq
If in (\ref{expression-260}) we interchange $d(u_1)$ with $d(u_2)$, and $d(v_1)$ with $d(v_2)$,
then we obtain an expression identical to (\ref{expression-150}), which was shown to be negative.

\smallskip
\noindent
{\bf Subcase~$3.2.$}  All children vertices of $u_1$, except $v_1$, are of degree $5$.

\noindent
Apply the following transformation $\mathcal{T}_{32}$ to $G$:
from $v_1$, and three children vertices of $u_2$ (named by $v_2$, $v_3$ and $v_4$), 
cut an adjacent pendant path $P_2$, form a $B_3^*$-branch
and attach it to $v_3$.
An example of this case with an illustration of the transformation $\mathcal{T}_{32}$ is given in Figure~\ref{fig-case3_2}.
\begin{figure}[h]
\begin{center}
\includegraphics[scale=0.750]{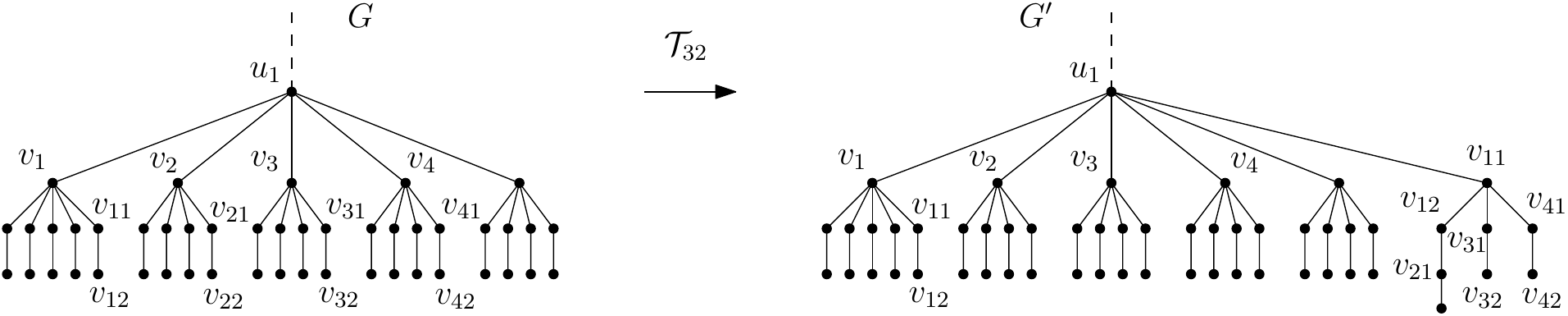}
\caption{Transforamation $\mathcal{T}_{32}$ from Subcase~$3.2.$}
\label{fig-case3_2}
\end{center}
\end{figure}
%
After this transformation the degrees of $v_1$,$v_2$,  $v_3$ and $v_4$ decrease by one, 
the degree of $u_1$ increases by one, and the degree of one child vertex of $v_1$, denoted by $v_{11}$, increases by two. 
The degrees of the rest of the vertices of $G$  remain unchanged. 
The change of the ABC index here is
\beq \label{expression-280}
 && -f(d(u_1), d(v_1)) + f(d(u_1), d(v_1)-1)   +3(-f(d(u_1), 5) + f(d(u_1), 4) ) \nonumber \\ 
 &&  -f(d(v_1), 2) + f(d(u_1), 4). 
\eeq

\noindent
By Proposition~\ref{appendix-pro-030},  $-f(d(u_1), d(v_1)) + f(d(u_1), d(v_1)-1)$ 
decreases in $d(v_1)$,
therefore, (\ref{expression-280}) is maximal for $d(v_1)=6$.
Thus,
\beq \label{expression-240}
g(d(u_1)) & = & -f(d(u_1), 6) + f(d(u_1), 5)  +3(-f(d(u_1), 5) + f(d(u_1), 4) ) \nonumber \\ 
&& + f(d(u_1), 4)  -\sqrt{\frac{1}{2}}  \nonumber
\eeq
is an upper bound on (\ref{expression-280}).
The first derivative of $g(d(u_1)) $
after a simplification is
\beq \label{expression-250}
\ds \frac{d g(d(u_1))}{d d(u_1)} = 
\frac{-\frac{30}{\sqrt{\frac{2+u_1}{u_1}}}+\frac{9 \sqrt{5}}{\sqrt{\frac{3+u_1}{u_1}}}+\frac{5 \sqrt{6}}{\sqrt{\frac{4+u_1}{u_1}}}}{15 u_1^2}, \nonumber
\eeq
which is negative for $u_1 \leq 6.27567$, and positive for $u_1 > 6.27567$, i.e.,
$g(d(u_1)$ decreases in $u_1$ when $u_1 =6$, and increases in $u_1$ when $u_1 \in [7,\infty)$.
Thus the upper bound on $g(d(u_1))$, and therefore on (\ref{expression-280}) is
$$
\max \left(g(6), \lim_{d(u_1) \to \infty} g(d(u_1))\right) = \max \left(-0.0201971,  -0.00978226 \right) = -0.00978226,
$$
\noindent
Hence, we have shown that the change of the ABC index also in this case is negative.
This conclude the proof of the theorem.
\end{proof}

\noindent
The next result  consider a (non)coexistence of some types of $B_k$-branches that have a common parent vertex. 
The result will be used in the proof of Theorem~\ref{thm-20}.

\begin{lemma} \label{lemma-10}
A minimal-ABC tree does not contain 
\begin{enumerate}
\item[(a)] a $B_1$-branch and a $B_4$-branch,
\item[(b)] a $B_2$-branch and a $B_4$-branch, 
\end{enumerate}
that have a common parent vertex.
\end{lemma}
\begin{figure}[h!]
\begin{center}
\includegraphics[scale=0.750]{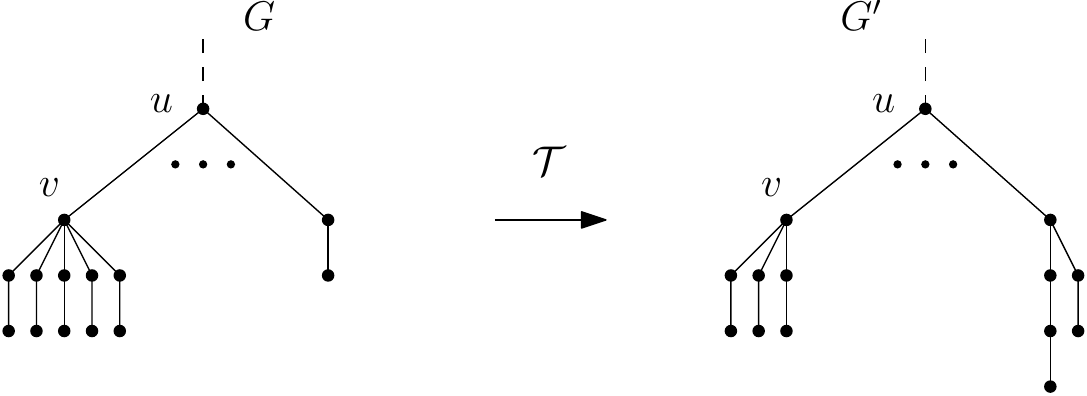}
\caption{Transforamations $\mathcal{T}_{a_1}$ and  $\mathcal{T}_{a_2}$ from Lemma~\ref{lemma-10}$(a)$.}
\label{fig-common_vertex_P2_B4}
\end{center}
\end{figure}
\begin{proof}

\smallskip
\noindent
$(a)$
Denote by $u$ the common vertex of the $B_1$-branch and the $B_4$-branch. If $d(u) \leq 241$ consider the transformation
$\mathcal{T}_{a_1}$ depicted in Figure~\ref{fig-common_vertex_P2_B4}. The change of ABC index after applying 
this transformation is
\beq \label{eq-lemma-B4-10}
ABC(G') - ABC(G) = -f(d(u), 5) + f(d(u), 3) -f(d(u), 2) + f(d(u), 3).
\eeq 
The first derivative of the above expression after a simplification is
\beq \label{eq-lemma-B4-15}
\frac{9 \sqrt{5} \sqrt{\frac{1+d(u)}{d(u)}}-10 \sqrt{3} \sqrt{\frac{3+d(u)}{d(u)}}}{30 d(u)^2 \sqrt{\frac{1+d(u)}{d(u)}} \sqrt{\frac{3+d(u)}{d(u)}}}.
\eeq
It is easy to show that $9 \sqrt{5} \sqrt{(1+d(u))/d(u)}-10 \sqrt{3} \sqrt{(3+d(u))/d(u)}$, 
is positive for $d(u) \geq 2$.
Hence, (\ref{eq-lemma-B4-15}) is positive also, and the difference $ABC(G') - ABC(G)$ from (\ref{eq-lemma-B4-10}) is increasing function in $d(u)$.
It is equal to zero for $d(u)=242$.

For $d(u) \geq 242$ consider the transformation $\mathcal{T}_{a_2}$ depicted in Figure~\ref{fig-common_vertex_P2_B4}.
Let $x$ be a parent vertex of $u$, and $y_i$, $i=1, \dots , d(u)-3$, are the children 
vertices of $u$ different than $v$ and $w$.
Then, the change of ABC index after applying $\mathcal{T}_{a_2}$ is
\beq \label{eq-lemma-B4-20}
ABC(G') - ABC(G) &= &-f(d(u), 5) + f(d(u)-1, 6) -f(d(u), d(x)) + f(d(u)-1, d(x)) \nonumber \\
                             & & +\sum_{i=1}^{d(u)-3}(-f(d(u), y_i) + f(d(u)-1, y_i)).
\eeq 
If $u$ is the root vertex of $G$, then the change of the 
ABC index after applying $\mathcal{T}_{a_2}$ is
\beq \label{eq-lemma-B4-30}
ABC(G') - ABC(G) = -f(d(u), 5) + f(d(u), 6)   +\sum_{i=1}^{d(u)-2}(-f(d(u), y_i) + f(d(u)-1, y_i)).
\eeq
By Propostion~\ref{appendix-pro-040},  $-f(d(u), d(x)) + f(d(u)-1, d(x))$ (resp. $-f(d(u), y_i) + f(d(u)-1, y_i)$)  increases in $x$ (resp. $y_i$), 
for $i=1, \dots , d(u)-2$.
Since $x \geq y_i$ it follows that the difference (\ref{eq-lemma-B4-20}) is at least so large as the difference  (\ref{eq-lemma-B4-30}).
To show that both differences are negative, it suffices to show that the difference  $ABC(G') - ABC(G)$ in (\ref{eq-lemma-B4-20}) is negative.
We have
\beq \label{eq-lemma-B4-40}
ABC(G') - ABC(G) &\leq &-f(d(u), 5) + f(d(u)-1, 6) \nonumber \\
                             & & +(d(u)-2)(-f(d(u), d(x)) + f(d(u)-1, d(x))) \nonumber \\
                             &=& f(d(x),d(u)) .
\eeq 
Because   $-f(d(u), d(x)) + f(d(u)-1, d(x))$ increases in $d(x)$, it follows that the difference (\ref{eq-lemma-B4-40}) is largest 
when $d(x) \to \infty$.
The partial derivative of (\ref{eq-lemma-B4-40}) with respect to $d(u)$, when  $d(x) \to \infty$, after a simplification, is
\beq \label{eq-lemma-B4-50}
&&\ds \lim_{d(x) \to \infty} \frac{\partial f(d(x),d(u))}{\partial{u}} =  \nonumber \\
&&\frac{1}{30} \left(15 \left(\frac{1}{u-1}\right)^{3/2} u-15 \left(\frac{1}{u}\right)^{3/2} (2+u)-\frac{10 \sqrt{6}}{(u-1)^{3/2} \sqrt{u+3} }+\frac{9 \sqrt{5}}{ u^{3/2} \sqrt{u+3} }\right)  \nonumber
\eeq
Straightforward verification shows that
$$
15 \left(\frac{1}{u-1}\right)^{3/2} u-15 \left(\frac{1}{u}\right)^{3/2} (2+u) <0
$$
for  $u \geq 5$.
Also, it is straightforward to show that 
$$
-\frac{10 \sqrt{6}}{(u-1)^{3/2} \sqrt{u+3} }+\frac{9 \sqrt{5}}{ u^{3/2} \sqrt{u+3} } \leq 0
$$
 for  any real $u$.
Therefore, for  $d(u) \geq 5$,
it follows that the difference   (\ref{eq-lemma-B4-40}) (with $d(x) \to \infty$)  is decreasing function with respect to $d(u)$,
and it is negative for $d(u) \geq 170$ (notice that it was sufficient to show that the
difference  (\ref{eq-lemma-B4-40}) is negative for $d(u) \geq 242$).

Thus, we have proven that the configuration $(a)$ from this lemma does not belong to a minimal-ABC tree,
since by the transformations $\mathcal{T}_{a_1}$ and $\mathcal{T}_{a_2}$ we have obtain a tree
$G'$ with smaller ABC index than $G$.

\smallskip
\noindent
$(b)$
Assume that a tree with minimal ABC index contains a configuration with a $B_2$-branch and a $B_4$-branch
having a common parent vertex $u$.
Then, apply the transformation $\mathcal{T}_{b}$ depicted in Figure~\ref{fig-common_vertex_B2_B4}.
After this transformation the change of the ABC index of $G$ is
\beq \label{eq-lemma-B4-70}
ABC(G') - ABC(G) &= &-f(d(u), 5) + f(d(u), 4) -f(d(u), 3) + f(d(u), 4).
\eeq 
\begin{figure}[h!]
\begin{center}
\includegraphics[scale=0.750]{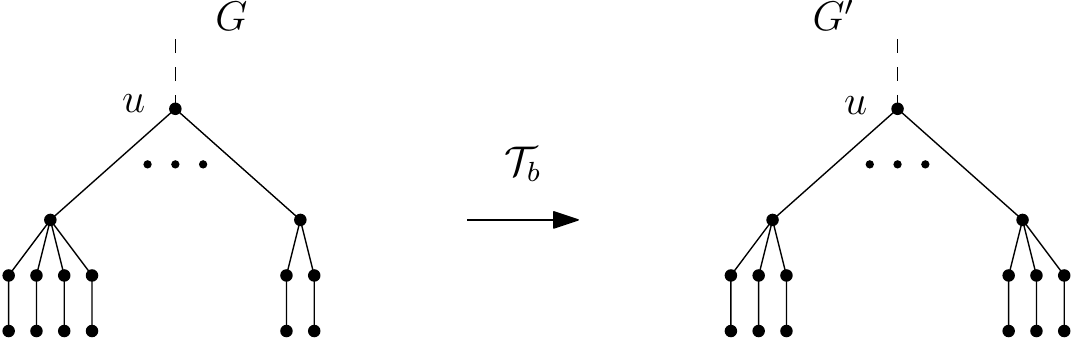}
\caption{Transforamation $\mathcal{T}_{b}$ from  Lemma~\ref{lemma-10}$(b)$.}
\label{fig-common_vertex_B2_B4}
\end{center}
\end{figure}
By Proposition~\ref{appendix-pro-030}, it follows that 
$-f(d(u), 5) + f(d(u), 4) < -f(d(u), 4) + f(d(u), 3) $.
From here, it follows that the difference $(\ref{eq-lemma-B4-70})$ is negative.
Hence, we obtain a contradiction to the initial assumption that $G$ is a tree with minimal ABC index.
\end{proof}

\smallskip
\noindent
Next, we present an upper bound on the number of $B_4$-branches that a graph with minimal ABC index can have.

\begin{te} \label{thm-20}
A  minimal-ABC tree does not contain more than four $B_4$-branches.
\end{te}

\begin{proof}
Assume that a tree $G$ with minimal ABC index has more than four  $B_4$-branches.
Consider the last five $B_4$-branches with respect to the breadth-first search of $G$.
Denote the corresponding root vertices of those branches by $v_1$, $v_2$, $v_3$, $v_4$ and $v_5$.
Assume that $d(v_1) \geq d(v_2) \geq d(v_3) \geq d(v_4) \geq d(v_5)$. 
Note that $v_1$, $v_2$, $v_3$, $v_4$ and $v_5$ can
have a common parent vertex, or can have two different parent vertices.
With respect to that, we consider two cases.

\bigskip
\noindent
{\bf Case~$1.$} $v_1$, $v_2$, $v_3$, $v_4$ and $v_5$  have two different parent vertices.

\smallskip
\noindent
Denote these vertices by $u_1$ and $u_2$.
Assume that $d(u_1) \geq d(u_2)$ and that $u_1$ is a parent vertex of $x$ vertices 
among  $v_1$, $v_2$, $v_3$, $v_4$ and $v_5$, where $1 \leq x < 5$.
Let $u$ be a parent vertex of $u_1$, and $y_i$, $i=1, \dots, d(u_2)-(5-x)-1$ the children vertices of $u_2$ that are not
in $\{ v_1, v_2, v_3, v_4, v_5 \}$.
Apply the following transformation $\mathcal{T}_{1}$: from each of $v_1$, $v_2$, $v_3$, $v_4$ and $v_5$
cut an adjacent pendant path $P_2$, form a $B_3^{**}$-branch and attached it to $u_2$.
An illustration, when $x=4$, is given in Figure~\ref{thm2case1}.
\begin{figure}[h]
\begin{center}
\includegraphics[scale=0.750]{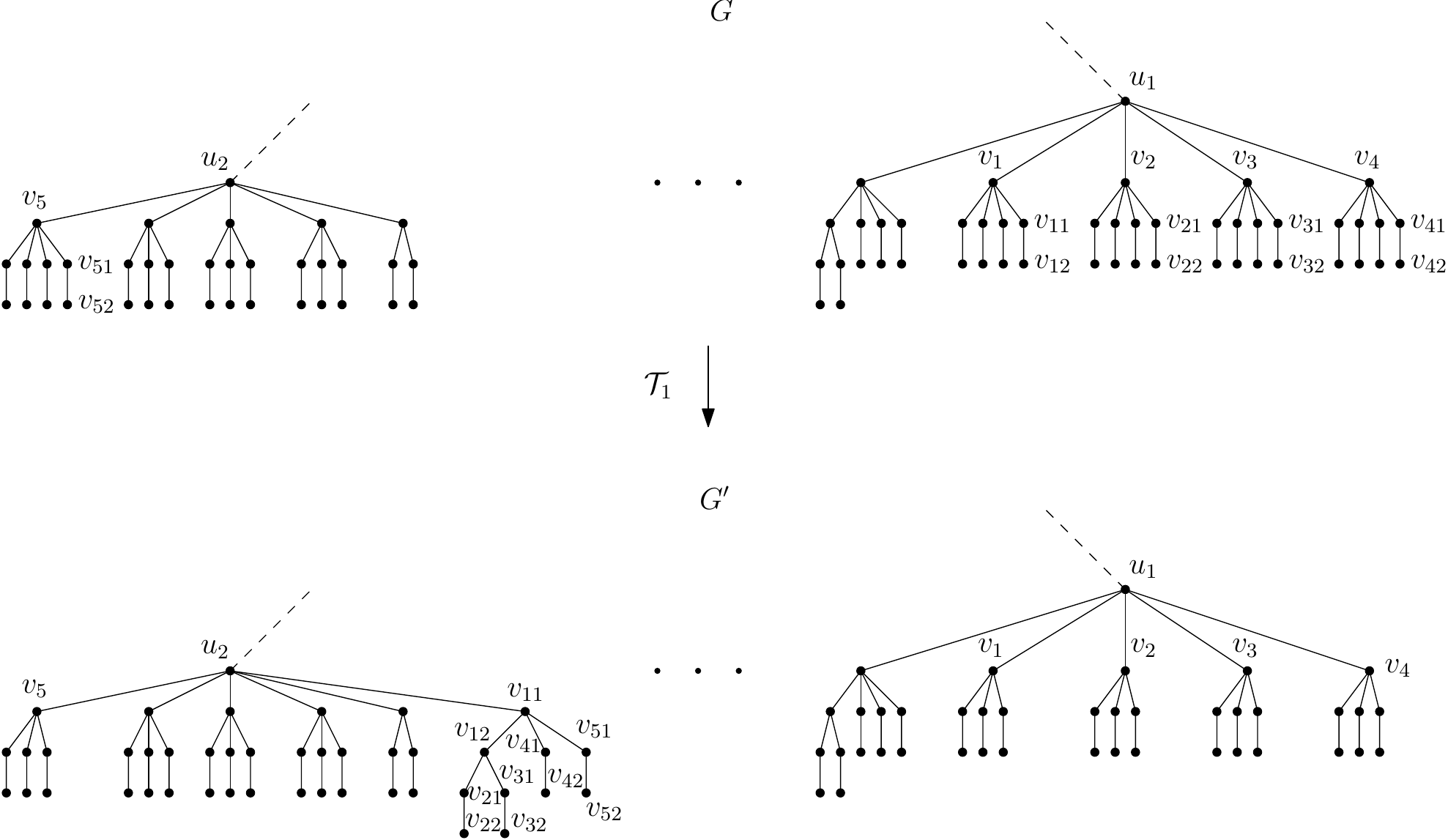}
\caption{Transformation $\mathcal{T}_1$ from Case~$1$.}
\label{thm2case1}
\end{center}
\end{figure}
After applying $\mathcal{T}_{1}$ the degrees of $v_1$, $v_2$, $v_3$, $v_4$ and $v_5$ decrease by one, while
the degree of $u_2$ increases by one. The degrees of the rest of the vertices of $G$  remain unchanged.
After this transformation, the change of the ABC index between $u_1$ and $x$ of its children vertices, that are
roots of $B_4$-branches, is $x(-f(d(u_1), 5) + f(d(u_1), 4))$.
The change of the ABC index between $u_2$ and $5-x$ of its children, that are
roots of $B_4$-branches, is $ (5-x)(-f(d(u_2), 5) + f(d(u_2)+1, 4))$.
$\sum_{i=1}^{d(u_2)-(5-x)-1}(-f(d(u_2), y_i) + f(d(u_2)+1, d(y_i))) $ is the change of the ABC index
caused by the rest of the children vertices of $u_2$ and $u_2$ itself, while
$-f(d(u_2), d(u)) + f(d(u_2)+1, d(u))$  is the change of the ABC index
caused by $u_2$ and its parent vertex.
Finally,  the change caused by attaching the $B_3^{**}$-branch to $u_2$ is
$-f(d(v_1),2) +  f(d(u_2)+1,4) - f(2,1)  + f(4,3)$.
Thus, the total change of the ABC index after applying $\mathcal{T}_{1}$ is
\beq \label{thm2-10-1}
ABC(G') - ABC(G) &=&
x(-f(d(u_1), 5) + f(d(u_1), 4)) \nonumber\\
&&+ (5-x)(-f(d(u_2), 5) + f(d(u_2)+1, 4)) \nonumber\\
 &&+\sum_{i=1}^{d(u_2)-(5-x)-1}(-f(d(u_2), d(y_i)) + f(d(u_2)+1, d(y_i))) \nonumber\\
 &&-f(d(u_2), d(u)) + f(d(u_2)+1, d(u))  \nonumber\\
 &&-f(d(v_1),2) +  f(d(u_2)+1,4) - f(2,1) + f(4,3)\nonumber\\
 &=&g(d(u_1), d(u_2), x). \nonumber
\eeq 
%
%


\noindent
By Proposition~\ref{appendix-pro-030}, $-f(d(u_1), 5) + f(d(u_1), 4)$ increases in $d(u_1)$, 
therefore the function  $g(d(u_1), d(u_2), x)$ reaches its maximum when $d(u_1) \to \infty$.
By Proposition~\ref{appendix-pro-020}, $-f(d(u_2), d(u)) + f(d(u_2)+1, d(u))$ and 
$-f(d(u_2), y_i) + f(d(u_2)+1, d(y_i))$ decrease in $d(u)$ and $d(y_i)$, respectively. 
It holds that $d(u) \geq d(u_2) \geq 5$. Thus, $-f(d(u_2), d(u_2)) + f(d(u_2)+1, d(u))$ 
has its maximum for $d(u)=d(u_2)$. The function  $-f(d(u_2), d(u_2)) + f(d(u_2)+1, d(u_2))$
increases in $d(u_2)$ and has a maximum of $0$.
Since $u_2$ is a parent vertex of $B_4$-branch,  by Lemma~ \ref{lemma-10}, 
$u_2$ cannot be a parent vertex of $B_1$-branch or $B_2$-branch.
Thus, $g(d(u_1), d(u_2), x)$ is maximal for $d(y_i)=4$, $i=1,\dots, d(u_2)-(5-x)-1$.
Next, we show that  
\beq \label{thm2-10-10}
g_1(d(x),d(u_2)) &=&
(5-x)(-f(d(u_2), 5) + f(d(u_2)+1, 4)) \nonumber\\
 &&+(d(u_2)-(5-x)-1)(-f(d(u_2), 4) + f(d(u_2)+1, 4)) \nonumber \\
 &&-f(d(u_2), d(u_2)) + f(d(u_2)+1, d(u_2))  \nonumber\\
 && +  f(d(u_2)+1,4) \nonumber
\eeq 
increases in $d(u_2)$.
Indeed, it can be verified that $\partial g_1(d(x),d(u_2))/\partial{u_2} \neq 0$, for $u_2 \in [4, \infty)$, and
$\partial g_1(d(x),d(u_2))/\partial{u_2}$ is always positive.
%
%
 Thus, we obtain
 \beq \label{thm2-10-20}
\ds ABC(G') - ABC(G) &<&
 \lim_{  \substack{
d(u_1) \to \infty \\
d(u_2) \to \infty
} } \left( x(-f(d(u_1), 5) + f(d(u_1), 4)) \right. \nonumber\\
&& \left.+ (5-x)(-f(d(u_2), 5) + f(d(u_2)+1, 4))  \right. \nonumber\\
&& \left.+(d(u_2)+x-6)(-f(d(u_2), 4) + f(d(u_2)+1, 4))  \right. \nonumber\\
&& \left. -f(d(u_2), d(u_2)) + f(d(u_2)+1, d(u_2))   \right) \nonumber\\
&& \left. -f(d(v_1),2) - f(2,1) +  f(d(u_2)+1,4) + f(4,3)  \right. \nonumber\\
 &=&-0.00478432,  \nonumber
\eeq
for $x=1,2, 3,4$.
Hence, we have shown that after applying the transformation $\mathcal{T}_{1}$, the ABC index strictly decreases.

\smallskip
\noindent
{\bf Case~$2.$} $v_1$, $v_2$, $v_3$, $v_4$ and $v_5$  have a common parent vertex $u_1$.

\smallskip
\noindent
Here, we apply a similar transformation $\mathcal{T}_{2}$ to the transformation $\mathcal{T}_{1}$ from the previous case: from each of $v_1$, $v_2$, $v_3$, $v_4$ and $v_5$
we cut an adjacent pendant path $P_2$, form a $B_3^{**}$-branch and attached it to $u_1$.
An illustration is given in Figure~\ref{thm2case2}.
\begin{figure}[h]
\begin{center}
\includegraphics[scale=0.750]{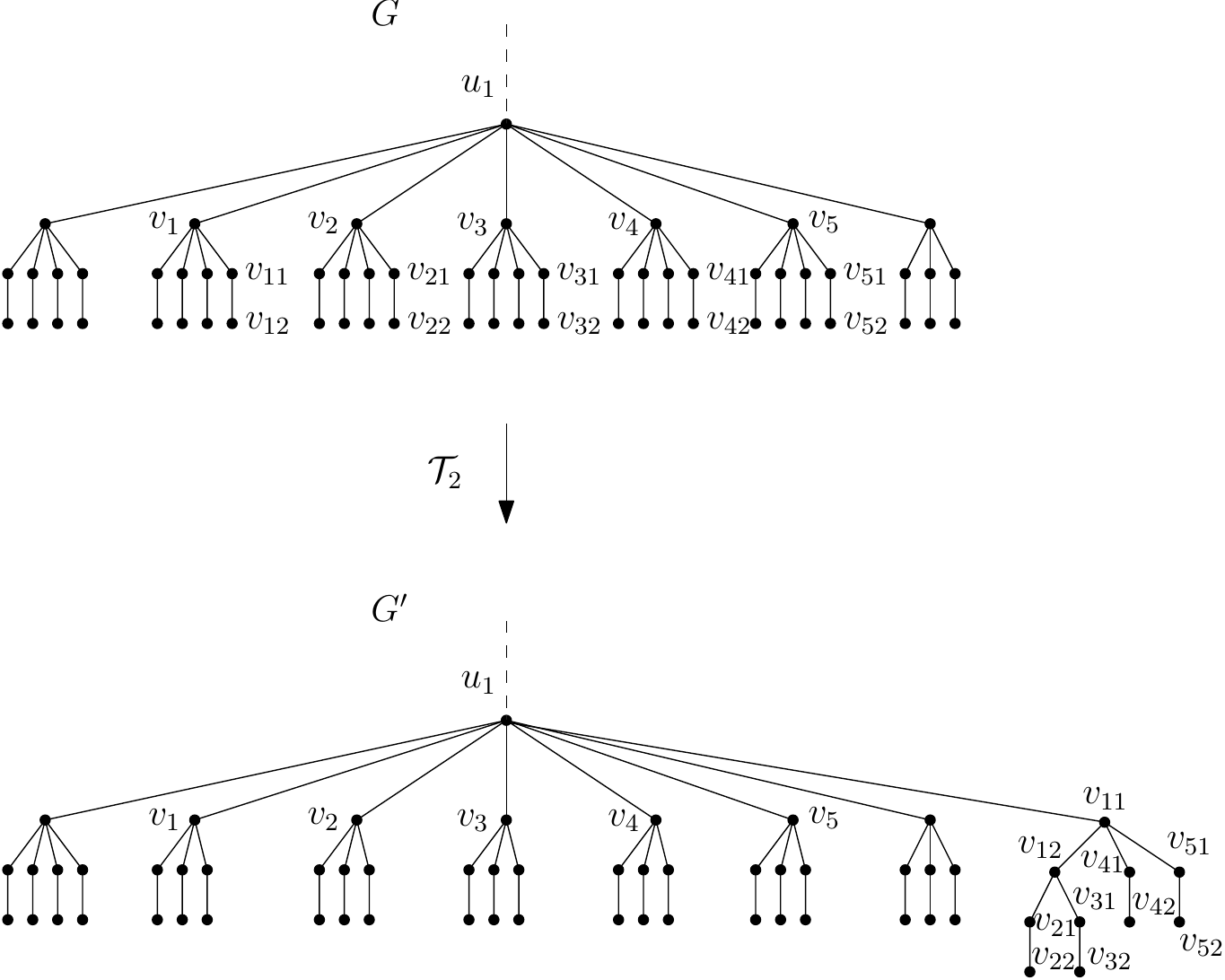}
\caption{Transformation $\mathcal{T}_2$ from Case~$2$.}
\label{thm2case2}
\end{center}
\end{figure}
After applying $\mathcal{T}_{2}$ the degrees of $v_1$, $v_2$, $v_3$, $v_4$ and $v_5$ decrease by one, while
the degree of $u_1$ increases by one. The degrees of the rest of the vertices of $G$  remain unchanged.
We distinguish two further cases with respect the $u_2$.

\smallskip
\noindent
{\bf Subcase~$2.1.$} $u_1$ is the root of $G$.

\smallskip
\noindent
In this case, after applying  $\mathcal{T}_{2}$, the change of the ABC index between $u_1$ and its children vertices, that are
roots of $B_4$-branches, is $5(-f(d(u_1), 5) + f(d(u_1), 4))$, while
the change of the ABC index
caused by the rest of the children vertices of $u_2$ and $u_2$ itself is
$\sum_{i=1}^{d(u_1)-5}(-f(d(u_1), y_i)$ $+ f(d(u_1)+1, d(y_i))) $.
The change caused by attaching the $B_3^{**}$-branch to $u_1$ is
$-f(d(v_1),2) +  f(d(u_1)+1,4) - f(2,1)  + f(4,3)$.
Thus, the total change of the ABC index after applying $\mathcal{T}_{2}$ is
\beq \label{thm2-10}
ABC(G') - ABC(G) &=&
5(-f(d(u_1), 5) + f(d(u_1)+1, 4)) \nonumber\\
&&+ \sum_{i=1}^{d(u_1)-5}(-f(d(u_1), d(y_i)) + f(d(u_1)+1, d(y_i))) \nonumber\\
&&-f(d(u_1),2) +  f(d(u_1) +1,4) - f(2,1) + f(4,3). 
\eeq 
By Proposition~\ref{appendix-pro-020}, the expression $-f(d(u_1), y) + f(d(u_1)+1, y)$ is negative for $d(u_1) > 2$, and
$-f(d(u_1), y_i)$ $+$ $f(d(u_1)+1, d(y_i))$ decrease in  $d(y_i)$. 
Since $u_1$ is a parent vertex of $B_4$-branch,  by Lemma~ \ref{lemma-10}, 
$u_1$ cannot be a parent vertex of $B_1$-branch or $B_2$-branch.
Thus, the change of the ABC index is maximal for $d(y_i)=4$, $i=1,\dots, d(u_1)-5$.
Thus, with a further rearrangement of (\ref{thm2-10}), we obtain
\beq \label{thm2-20}
ABC(G') - ABC(G) & < &
-5f(d(u_1), 5) + 6f(d(u_1)+1, 4) -2f(2,1)  + f(4,3). 
\eeq 
We have that 
\beq \label{eq-thm2-30}
\ds \frac{d (-5f(d(u_1), 5) + 6f(d(u_1)+1, 4))}{d{u_1}} &=&
 \frac{3 \sqrt{5}}{2 u_1^2 \sqrt{\frac{3+u_1}{u_1}}}-\frac{3}{(1+u_1)^2 \sqrt{\frac{3+u_1}{1+u_1}}}. \nonumber
\eeq
A straightforward  verification shows that
\beq \label{eq-thm2-40}
 \frac{3 \sqrt{5}}{2 u_1^2 \sqrt{\frac{3+u_1}{u_1}}}-\frac{3}{(1+u_1)^2 \sqrt{\frac{3+u_1}{1+u_1}}} >0, \nonumber
\eeq
for every positive $u_1$.
It follows that $-5f(d(u_1), 5) + 6f(d(u_1)+1, 4)$ increases with $u_1$.
Thus, we obtain
\beq \label{thm2-60}
ABC(G') - ABC(G)  &<&
\lim_{u_1 \to \infty}-5f(d(u_1), 5) + 6f(d(u_1)+1, 4) - 2f(2,1)  + f(4,3) \nonumber \\
&=& -0.00478432. \nonumber
\eeq

\smallskip
\noindent
{\bf Subcase~$2.2.$} $u_1$ is not the root of $G$.

\smallskip
\noindent
Then,  $d(u_1) \geq 6$, and the change of the ABC index is
\beq \label{thm2-30}
ABC(G') - ABC(G) &=&
 -f(d(u_1), d(x)) + f(d(u_1)+1, d(x)) \nonumber\\
&& 5(-f(d(u_1), 5) + f(d(u_1)+1, 4)) \nonumber\\
&&+  \sum_{i=1}^{d(u_1)-6}(-f(d(u_1), d(y_i)) + f(d(u_1)+1, d(y_i)))\nonumber\\
&&-f(d(u_1),2) - f(2,1) +  f(d(u_1)+1,4) + f(4,3), \nonumber
\eeq 
where $x$ is a parent vertex of $u_1$.
%
%
Applying the same arguments as in Subcase~$2.1.$, we obtain
\beq \label{thm2-40}
ABC(G') - ABC(G) & < &
-5f(d(u_1), 5) + 6f(d(u_1)+1, 4) - 2f(2,1)  + f(4,3), \nonumber
\eeq 
which is identical with (\ref{thm2-20}) from Subcase~$2.1.$ Therefore, the
change of the ABC index after applying the transformation $\mathcal{T}_2$ is negative. 

Applying repeatedly above considered transformations $\mathcal{T}_1$ and $\mathcal{T}_2$, one can obtain a tree with at most
$4$ $B_4$-branches, that has smaller ABC index than the assumed minimal-ABC tree.
\end{proof}

%
%

%
\section[Appendix]{Appendix} \label{general}

Here we  present a collection of auxiliary results that were used in the proofs in the main text.
In the next propositions the function $f(x,y)$ is defined as in (\ref{eqn:000}).

\begin{pro}  \label{appendix-pro-020}
Let $g(x,y)=-f(x,y)+f(x+1,y)$, with real numbers $x, y \geq 2$.
Then, $g(x,y)$ is non-positive (strictly negative for $y > 2$) and increases in $x$ and decreases in $y$.
\end{pro}
\begin{proof}
First, we show that
\beq \label{eq-appendix-45}
-\sqrt{\frac{x+y-2}{x y}} + \sqrt{\frac{x+y-1}{(x+1) y}} \leq 0. 
\eeq
Indeed, after squaring (\ref{eq-appendix-45}) and further simplification, we obtain
$
2 - y \leq 0.
$
The equality holds for $y=2$.

\noindent
The first partial derivative of $g(x,y)$ with respect to $x$ is
\beq \label{eq-appendix-50}
\ds \frac{\partial g(x,y)}{\partial{x}} &=& 
\frac{1}{2 y}   \left( \frac{-2+y}{x^2 \sqrt{\frac{-2+x+y}{x y}}}+\frac{2-y}{(1+x)^2 \sqrt{\frac{-1+x+y}{(1+x) y}}} \right). \nonumber
\eeq
%
Applying simple algebraic transformations, one can transform
\beq \label{eq-appendix-60}
\ds \frac{-2+y}{x^2 \sqrt{\frac{-2+x+y}{x y}}} + \frac{2-y}{(1+x)^2 \sqrt{\frac{-1+x+y}{(1+x) y}}} > 0  \nonumber
\eeq
into
\beq \label{eq-appendix-70}
\ds (1+x)^3(x+y-1)>x^3(x+y-2),  \nonumber
\eeq
which holds for $x, y \geq 2$. Therefore, ${\partial g(x,y)}/{\partial x} > 0$, from which it follows that $g(x,y)$ increases in $x$.

\noindent
The first partial derivative of $g(x,y)$ with respect to $y$ is
\beq \label{eq-appendix-80}
\ds \frac{\partial g(x,y)}{\partial{y}} &=& 
\frac{1}{2 y^2}   \left( \frac{-2+x}{x \sqrt{\frac{-2+x+y}{x y}}}+\frac{1-x}{(1+x) \sqrt{\frac{-1+x+y}{(1+x) y}}} \right). \nonumber 
\eeq
The first partial derivative $\partial g(x,y)/ \partial{y}$ is negative if 
\beq \label{eq-appendix-90}
\frac{-2+x}{x \sqrt{\frac{-2+x+y}{x y}}}+\frac{1-x}{(1+x) \sqrt{\frac{-1+x+y}{(1+x) y}}} <0.
\eeq
After squaring  and algebraic rearranging of (\ref{eq-appendix-90}), we obtain
\beq \label{eq-appendix-100}
 (x+1)(x-2)^2(x+y-1)-x(x-1)^2(x+y-2) <0,  \nonumber 
\eeq
which is fulfilled for $x, y \geq 2$. Therefore, $\partial g(x,y)/ \partial{y}$ is negative, which implies that $g(x,y)$ decreases in $y$.
\end{proof}

\begin{pro}  \label{appendix-pro-030}
Let $g(x,y)=-f(x,y)+f(x,y-1)$, with real numbers $x, y \geq 2$.
Then, $g(x,y)$ is non-negative and increases in $x$ and decreases in $y$.
\end{pro}
\begin{proof}
Since $f(x,y)$ is symmetric function, it holds that $-g(x,y)=f(x,y)-f(x,y-1)=-f(y-1,x)+f(y,x)$.
By Proposition~\ref{appendix-pro-020}, $-g(x,y)$ is non-positive and increases in $y$, and decreases in $x$.
Thus, it follows that $g(x,y)$ is non-negative and increases in $x$ and decreases in $y$.
\end{proof}

\begin{pro}  \label{appendix-pro-010}
Let $g(x,y)=-f(x,y)+f(x+\Delta x,y-\Delta y)$, with  real numbers $x, y \geq 2$,  $\Delta x \geq 0$,  $0 \leq \Delta y < y$.
Then, $g(x,y)$ increases in $x$ and decreases in $y$.
\end{pro}
\begin{proof}
Let $g_1(x,y)=-f(x,y)+f(x+1,y)$ and $g_2(x,y)=-f(x,y)+f(x,y-1)$.
Then,
\beq \label{eq-appendix-105}
g(x,y )&=&g_1(x,y)+g_1(x+1,y) +\dots + g_1(x+\Delta x -1,y) \nonumber \\
&+& g_2(x+\Delta x ,y)  + g_2(x+\Delta x ,y-1) +\dots +g_2(x+\Delta x ,y-\Delta y +1). \nonumber
\eeq
By Propositons~\ref{appendix-pro-020} and \ref{appendix-pro-030}, all of the functions $g_1$ and $g_2$ in the above expression 
increase in $x$ and decrease in $y$.
Therefore, $g(x,y)$  also increases in $x$ and decreases in $y$.
\end{proof}

\begin{pro}  \label{appendix-pro-040}
Let $g(x,y)=-f(x,y)+f(x-1,y)$, with positive real numbers $x, y \geq 2$.
Then, $g(x,y)$ is non-negative and increases in $y$ and decreases in $x$.
\end{pro}
\begin{proof}
It holds that $g(x,y)=-f(y, x)+f(y,x-1))$.
By Proposition~\ref{appendix-pro-030}, $g(x,y)$ is non-negative and increases in $y$ and decreases in $x$.
\end{proof}

\begin{pro}  \label{appendix-pro-050}
Let $g(x,k)=k\left(  -f(d(x), 6) + f(d(x)+1, 5) \right)  + f(d(x)+1, 3)$, with positive real numbers $x, k \geq 2$.
Then, $g(x,k)$ increases in $x$.
\end{pro}
\begin{proof}
Consider $g(x,k)$ as sum of two functions $g_1(x,k)=k\left(  -f(d(x), 6) + f(d(x)+1, 5) \right)$
and $g_2(x)=f(d(x)+1, 3)$.
The first derivative of $g_1(x,k)$ with respect to $x$ is
\beq \label{eq-appendix-110}
\ds \frac{\partial g_1(x,k)}{\partial{x}} 
&=&  \frac{1}{60} k \left(\frac{20 \sqrt{6}}{x^2 \sqrt{\frac{4+x}{x}}}-\frac{18}{(1+x)^2 \sqrt{\frac{4+x}{5+5 x}}}\right). \nonumber 
\eeq
It is easy to verify that 
$
20 \sqrt{6}/(x^2 \sqrt{(4+x)/x}) > 18/((1+x)^2 \sqrt{(4+x)/(5+5 x}) \nonumber 
$
is positive for any positive $x$ and $k$, from which follows that  $\ds \partial g_1(x,k) / \partial{x}$, or 
that $g_1(x,k)$ is increasing in $x$.
On the other hand, the function $g_2(x)$ decreases in $x$, because
\beq \label{eq-appendix-120}
 \frac{d  g_2(x)}{d x} &=&  -\frac{1}{2 \sqrt{3} (1+x)^2 \sqrt{\frac{2+x}{1+x}}} < 0.\nonumber
\eeq
To prove the claim of the proposition we will show that for $k \geq 2$,
$g_1(x,k)$ increases faster  in $x$ than $g_2(x)$ decreases  in $x$, or,
\beq \label{eq-appendix-130}
   \frac{1}{60} k \left(\frac{20 \sqrt{6}}{x^2 \sqrt{\frac{4+x}{x}}}-\frac{18}{(1+x)^2 \sqrt{\frac{4+x}{5+5 x}}}\right) 
   -\frac{1}{2 \sqrt{3} (1+x)^2 \sqrt{\frac{2+x}{1+x}}} > 0.
\eeq
After rearrangement, we obtain that (\ref{eq-appendix-130}) is equivalent to
\beq \label{eq-appendix-140}
      \frac{ 10 \sqrt{6} k}{x^2 \sqrt{\frac{4+x}{x}}} &>&
     \frac{ 9 \sqrt{5} k}{(1+x)^2 \sqrt{\frac{4+x}{1+ x}}}  +\frac{10 \sqrt{3}}{ (1+x)^2 \sqrt{\frac{2+x}{1+x}}}.
\eeq
Since
\beq \label{eq-appendix-150}
      \frac{ 9 \sqrt{5} k}{(1+x)^2 \sqrt{\frac{2+x}{1+ x}}}  +\frac{10 \sqrt{3}}{ (1+x)^2 \sqrt{\frac{2+x}{1+x}}} >
     \frac{ 9 \sqrt{5} k}{(1+x)^2 \sqrt{\frac{4+x}{1+ x}}}  +\frac{10 \sqrt{3}}{ (1+x)^2 \sqrt{\frac{2+x}{1+x}}}, \nonumber
\eeq
to prove (\ref{eq-appendix-140}), it sufices to prove
\beq \label{eq-appendix-160}
      \frac{ 10 \sqrt{6} k}{x^2 \sqrt{\frac{4+x}{x}}} &>&
      \frac{ 9 \sqrt{5} k + 10 \sqrt{3}}{(1+x)^2 \sqrt{\frac{2+x}{1+ x}}}. 
\eeq

\noindent
Indeed, after squaring the both sides of (\ref{eq-appendix-160}) and performing a simple algebraic transformation, we obtain that (\ref{eq-appendix-160}) is equivalent to
\beq \label{eq-appendix-190}
15 \left(80 k^2+280 k^2 x+360 k^2 x^2-20 x^3-24 \sqrt{15} k x^3+92 k^2 x^3-5 x^4-6 \sqrt{15} k x^4+13 k^2 x^4\right)  > 0,  \nonumber
\eeq
which holds for $k \geq 2$ and $x \geq 0$.
\end{proof}

\end{document}